\documentclass{llncs}

\usepackage{amsmath,amsfonts,amssymb}
\usepackage[Symbol]{upgreek}
\usepackage[normalem]{ulem}
\usepackage{verbatim}
\usepackage{tikz}
\usepackage{wrapfig}
\usepackage{paralist}

\usepackage{cite}
\usepackage{hyperref}
\usepackage{thm-restate}

\let\doendproof\endproof
\renewcommand\endproof{~\hfill\qed\doendproof}

\makeatletter
\def\infrule#1#2#3{\@ifnextchar[{\@infrule{#1}{#2}{#3}}{\@infrule{#1}{#2}{#3}[*]}}

\def\@infrule#1#2#3[#4]{
 \par\bigbreak
 \vtop{
  \hangindent3em\hangafter2\leavevmode\null
  \textsf{#1:}\kern.5em\textbf{#2}\\[\smallskipamount]
  \ifx*#4 
   \null\qquad$#3$
  \else 
   \setbox30=\hbox{\qquad$#3$\qquad#4}
   \ifdim\wd30>\hsize
    \null\qquad$#3$\par\kern-\parskip\smallskip#4
   \else
    \null\qquad$#3$\qquad#4
   \fi
  \fi
 }%
} \makeatother

\allowdisplaybreaks
\newcommand{\moregen}{\mathrel{\leq\kern-2.9pt\raise0.8pt\hbox{\llap{$\cdot$}}}}

 \newcommand{\dom}{\mathrm{Dom}}
 \newcommand{\ran}{\mathrm{Ran}}

 \newcommand{\Lra}{\Longrightarrow}
 
 \newcommand{\vars}{\mathrm{Vars}}

 \newcommand{\frN}{\mathfrak N}

 \newcommand{\id}{\mathit{Id}}

\newcommand{\swap}[2]{\ensuremath{(#1\, #2)}}

\newcommand{\idp}{\mathit{Id}}
\newcommand{\head}{\mathrm{Head}}
\newcommand{\perm}{\approx}
\newcommand{\permap}{{\mkern 1mu\cdot \mkern 1mu}}
\newcommand{\permef}{\bullet}

\newcommand{\ids}{\varepsilon}

\newcommand\nomh{{\sf h}}
\newcommand\pos{{\sf p}}

\newcommand\ppi{\uppi}
\newcommand\prho{\uprho}
\newcommand\ptau{\uptau}
\newcommand\pmu{\upmu}

\newcommand{\dotcup}{\ensuremath{\mathop{\mathaccent\cdot\cup}}}
\newcommand{\np}[2]{\langle #1, \allowbreak #2\rangle}

\newcommand{\msys}{\mathfrak{E}}

\newcommand{\mran}{A}

\newcommand{\mproblem}{E}

\newcommand{\mnabla}{\nabla}

\newcommand{\msusp}{\ppi}

\newcommand{\fresh}[1]{{\acute{#1}}}
\newcommand{\len}[1]{{\|#1\|}}

\newcommand{\vcard}[1]{\|#1\|_{_{\vars}}}

\newcommand{\Atoms}{A}
\newcommand{\atoms}{\mathrm{Atoms}}

\newcommand{\nablai}{{\nabla}}
\newcommand{\nablag}{\Gamma}

\newcommand{\fe}{{\sf FC}}

\newcommand{\abscard}[1]{\|#1\|_{_{\sf Abs}}}

\newcommand{\unused}{\mathit{fresh}}

\newcommand{\fatoms}{\mathrm{FA^{\text{-s}}}}
\newcommand{\fatom}{\mathrm{FA}}

\title{Nominal Anti-Unification}

\titlerunning{Nominal Anti-Unification}

\author{Alexander Baumgartner\inst{1} \and Temur Kutsia\inst{1} \and Jordi Levy\inst{2} \and Mateu Villaret\inst{3}}

\institute{Research Institute for Symbolic Computation\\
Johannes Kepler University Linz, Austria\\
\email{\{abaumgar,kutsia\}@risc.jku.at} \and
Artificial Intelligence Research Institute (IIIA)\\
  Spanish Council for Scientific Research (CSIC), 
  Barcelona, Spain\\
  \email{levy@iiia.csic.es} \and
Departament d'Inform\`{a}tica i Matem\`{a}tica Aplicada (IMA)\\
  Universitat de Girona (UdG),
  Girona, Spain.\\
  \email{villaret@ima.udg.edu}
}

\begin{document}
\maketitle

\begin{abstract}
We study nominal anti-unification, which is concerned with computing
least general generalizations for given terms-in-context. In general,
the problem does not have a least general solution, but if the set of
atoms permitted in generalizations is finite, then there exists a
least general generalization which is unique modulo variable renaming
and $\alpha$-equivalence. We present an algorithm that computes
it. The algorithm relies on a subalgorithm that constructively decides
equivariance between two terms-in-context. We prove soundness and
completeness properties of both algorithms and analyze their
complexity.  Nominal anti-unification can be applied to problems were
generalization of first-order terms is needed (inductive learning,
clone detection, etc.), but bindings are involved.
\end{abstract}

\section{Introduction}

Binders are very common in computer science, logic, mathematics,
linguistics. Functional abstraction $\lambda$, universal quantifier
$\forall$, limit $\lim$, integral $\int$ are some well-known examples
of binders. To formally represent and study systems with binding,
Pitts and
Gabbay~\cite{DBLP:conf/lics/GabbayP99,DBLP:journals/fac/GabbayP02,GabbayThesis}
introduced nominal techniques, based on the idea to give explicit
names to bound entities. It makes a syntactic distinction between
\emph{atoms,} which can be bound, and \emph{variables,} which can be
substituted. This approach led to the development of the theory of
nominal sets, nominal logic, nominal algebra, nominal rewriting,
nominal logic programming, etc.

Equation solving between nominal terms (maybe together with freshness
constraints) has been investigated by several authors, who designed
and analyzed algorithms for nominal
unification~\cite{DBLP:journals/tcs/UrbanPG04,
  DBLP:conf/rta/LevyV08, DBLP:conf/rta/LevyV10,
  DBLP:journals/tocl/LevyV12, DBLP:journals/tcs/CalvesF08,
  CalvesThesis}, nominal matching~\cite{DBLP:journals/jcss/CalvesF10},
equivariant unification~\cite{DBLP:journals/jar/Cheney10}, and permissive
nominal unification~\cite{dowek09,DBLP:journals/igpl/DowekGM10}. However, in
contrast to unification, its dual problem, anti-unification, has not
been studied for nominal terms previously.

The anti-unification problem for two terms $t_1$ and $t_2$ is
concerned with finding a term $t$ that is more general than the
original ones, i.e., $t_1$ and $t_2$ should be substitutive instances
of $t$. The interesting generalizations are the least general ones,
which retain the common structure of $t_1$ and $t_2$ as much as
possible. Plotkin~\cite{Plotkin70} and Reynolds~\cite{Reynolds70}
initiated research on anti-unification in the 1970s, developing
generalization algorithms for first-order terms. Since then,
anti-unification has been studied in various theories, including some
of those with binding constructs: calculus of
constructions~\cite{DBLP:conf/lics/Pfenning91},
$M\lambda$~\cite{DBLP:conf/icml/FengM92}, second-order lambda calculus
with type variables~\cite{DBLP:journals/amai/LuMHH00}, simply-typed
lambda calculus where generalizations are higher-order
patterns~\cite{DBLP:conf/rta/BaumgartnerKLV13}, just to name a few.

The problem we address in this paper is to compute generalizations for
nominal terms. More precisely, we consider this problem for nominal
\emph{terms-in-context}, which are pairs of a freshness context and a
nominal term, aiming at computing their least general generalizations
(lgg). However, it turned out that without a restriction, there is no
lgg for terms-in-context, in general. Even more, a \emph{minimal}
complete set of generalizations does not exist. This is in sharp
contrast with the related problem of anti-unification for higher-order
patterns, which always have a single
lgg~\cite{DBLP:conf/rta/BaumgartnerKLV13}. The reason is one can make
terms-in-context less and less general by adding freshness constraints
for the available (infinitely many) atoms, see
Example~\ref{exmp:nullary}.  Therefore, we restrict the set of atoms
which are permitted in generalizations to be fixed and finite. In this
case, there exists a single lgg (modulo $\alpha$-equivalence and
variable renaming) for terms-in-context and we design an algorithm to
compute it in $O(n^5)$ time.

There is a close relation between nominal and higher-order pattern
unification: One can be translated into the other by the
solution-preserving translation defined in \cite{cheney2005,
  DBLP:conf/rta/LevyV08, DBLP:journals/tocl/LevyV12} or the
translation defined for permisive terms
in~\cite{dowek09,DBLP:journals/igpl/DowekGM10}. We show that for
anti-unification, this method, in general, is not applicable.
Even if one finds conditions under which such a translation-based
approach to anti-unification works, due to complexity reasons it is
still better to use the direct nominal anti-unification algorithm
developed in this paper.

Computation of nominal lgg's requires to solve the equivariance
problem: Given two terms $s_1$ and $s_2$, find a permutation of atoms
which, when applied to $s_1$, makes it $\alpha$-equivalent to $s_2$
(under the given freshness context). This is necessary to guarantee
that the computed generalization is \emph{least} general. For
instance, if the given terms are $s_1=f(a,b)$ and $s_2=f(b,a)$, where
$a,b$ are atoms, the freshness context is empty, and the atoms
permitted in the generalization are $a,b,$ and $c$, then the
term-in-context $\np{\{c\#X, c\#Y\}}{f(X,Y)}$ generalizes
$\np{\emptyset}{s_1}$ and $\np{\emptyset}{s_2}$, but it is not their
lgg. To compute the latter, we need to reflect the fact that
generalizations of the atoms are related to each other: One can be
obtained from the other by swapping $a$ and $b$. This leads to an lgg
$\np{\{c\#X\}}{f(X,\swap{a}{b}\permap X)}$. To compute the permutation
$\swap{a}{b}$, an equivariance problem should be solved. Equivariance is
already present in $\alpha$-Prolog~\cite{alphaprolog} and 
Isabelle~\cite{isabelle}.  We develop a rule-based algorithm
for equivariance problems, which computes in quadratic time the
justifying permutation if the input terms are equivariant, and fails
otherwise. 

Both anti-unification and equivariance algorithms are implemented in the anti-unification algorithm library~\cite{DBLP:conf/jelia/BaumgartnerK14} and can be accessed from

\url{http://www.risc.jku.at/projects/stout/software/}.

Various variants of anti-unification, such as first-order,
higher-order, or equational anti-unification have been used in
inductive logic programming, logical and relational
learning~\cite{DeRaedt2008}, reasoning by
analogy~\cite{DBLP:conf/ausai/KrumnackSGK07}, program
synthesis~\cite{DBLP:books/sp/Schmid03}, program
verification~\cite{DBLP:journals/amai/LuMHH00}, etc. Nominal
anti-unification can, hopefully, contribute in solving similar
problems in nominal setting or in first-order settings where bindings play
an important role.

In this paper, we mainly follow the notation
from~\cite{DBLP:journals/tocl/LevyV12}.

\section{Nominal Terms}
\label{sect:nominal}
In \emph{nominal signatures}
we have \emph{sorts of atoms} (typically $\nu$) and \emph{sorts of
  data} (typically $\delta$) as disjoint sets. \emph{Atoms} (typically
$a, b, \ldots$) have one of the sorts of atoms. \emph{Variables}
(typically $ X, Y, \ldots$) have a sort of atom or a sort of data,
i.e. of the form $\nu \mid \delta$. In nominal terms, variables can
be instantiated and atoms can be bound. Nominal function symbols
(typically $ f, g, \ldots$) have an arity of the form $\tau_1 \times
\cdots\times \tau_n \rightarrow \delta$, where $\delta$ is a sort of
data and $\tau_i$ are sorts given by the grammar $\tau ::= \nu \mid
\delta \mid \langle \nu\rangle \tau$ . Abstractions have sorts of the
form $\langle \nu\rangle \tau.$

A \emph{swapping} $\swap{a}{b}$ is a pair of atoms of the same
sort. A \emph{permutation} is a (possibly empty) sequence of swappings. We use upright Greek letters (e.g., $\ppi, \prho$) to denote permutations.
\emph{Nominal terms} (typically $ t, s, u, r, q, \ldots$) are given by the
grammar:
\[ t ::=  f(t_1 , \ldots , t_n ) \mid a \mid a.t \mid \ppi \permap  X\]
where $ f$ is an $n$-ary function symbol, $a$ is an atom, $\ppi$ is a
permutation, and $ X$ is a variable. They are called respectively
\emph{application, atom, abstraction,} and \emph{suspension}. The
sorts of application and atomic terms are defined as usual, the sort
of $a.t$ is $\langle \nu\rangle \tau$ where $\nu$ is the sort of $a$
and $\tau$ is the sort of $t$, and the sort of $\ppi \permap X$ is the
one of $X$.

The \emph{inverse} of a permutation $\ppi=\swap{a_1}{b_1} \ldots \swap{a_n}{b_n}$ is the permutation $\swap{a_n}{b_n} \ldots \swap{a_1}{b_1}$, denoted by $\ppi^{-1}$. The empty permutation is denoted by $\idp$. 
The \emph{effect of a swapping} over an atom is defined by
$\swap{a}{b} \permef a = b$, $\swap{a}{b}\permef b = a$ and
$\swap{a}{b}\permef c = c$, when $c \notin \{a, b\}$. It is extended
to the rest of terms: $\swap{a}{b}\permef f(t_1,\ldots,t_n) =
f(\swap{a}{b}\permef t_1,\ldots,\swap{a}{b}\permef t_n)$,
$\swap{a}{b}\permef (c.t) = \left ( \swap{a}{b}\permef c \right
).\left (\swap{a}{b} \permef t \right ) $, and $\swap{a}{b}\permef
\ppi \permap X = \swap{a}{b}\ppi \permap X$, where $\swap{a}{b}\ppi$
is the permutation obtained by concatenating $\swap{a}{b}$ and
$\ppi$. The \emph{effect of a permutation} is defined by
$\swap{a_1}{b_1} \ldots \swap{a_n}{b_n}\permef t =
\swap{a_1}{b_1}\permef \left (\swap{a_2}{b_2} \ldots
\swap{a_n}{b_n}\permef t\right )$. The effect of the empty permutation
is $\idp\permef t = t$. We extend it to suspensions and write $X$ as
the shortcut of $\id\permap X$.

The set of variables of a term $t$ is denoted by $\vars(t)$. A term
$t$ is called \emph{ground} if $\vars(t)=\emptyset$. The set of
\emph{atoms} of a term $t$ or a permutation $\ppi$ is the set of all
atoms which appear in it and is denoted by $\atoms(t)$, $\atoms(\ppi)$
respectively. For instance, $\atoms(f(a.g(a), \swap{b}{c}\permap X,
d)=\{a,b,c,d\}$. We write $\atoms(t_1,\dots,t_n)$ for the set
$\atoms(t_1)\cup\dots\cup\atoms(t_n)$.

The \emph{length} of a term $t$ is the number of all appearances of
atoms, variables, and function symbols in it and we denote it by
$\len{t}$, e.g., $\len{f(a.\swap{a}{b} \permap X,X,Y)}=7$. The number
of variable occurrences in a term $t$ is denoted by $\vcard{t}$, e.g.,
$\vcard{f(a. \swap{a}{b} \permap X,X,Y)}=3$.

  \begin{wrapfigure}[12]{r}{0.34\textwidth}
  \begin{tikzpicture}
  \node (f) at (3.7,4) {$f$};
  \node (fpos) at (4.1,4) {$\epsilon$};
  \node (aabs) at (3,3.5) {$a.$};
  \node (aabspos) at (3.4,3.5) {\footnotesize{$1$}};
  \node (babs) at (3,2.8) {$b.$};
  \node (babspos) at (3.55,2.8) {\footnotesize{$1.1$}};
  \node (g) at (3,2.1) {$g$};
  \node (gpos) at (3.65,2.1) {\footnotesize{$1.1.1$}};
  \node (susp) at (2,1.4) {$\swap{a}{b}\permap X$};
  \node (susppos) at (3.09,1.4) {\footnotesize{$1.1.1.1$}};
  \node (a) at (4,1.4) {$a$};
  \node (apos) at (4.7,1.4) {\footnotesize{$1.1.1.2$}};
  \node (h) at (4.4,3.5) {$h$};
  \node (hpos) at (4.8,3.5) {\footnotesize{$2$}};
  \node (c) at (4.4,2.8) {$c$};
  \node (cpos) at (4.9,2.8) {\footnotesize{$2.1$}};
  \draw[thick] (f) -- (aabs);
  \draw[thick] (f) -- (h);
  \draw[thick] (aabs) -- (babs);
  \draw[thick] (babs) -- (g);
  \draw[thick] (g) -- (susp);
  \draw[thick] (g) -- (a);
  \draw[thick] (h) -- (c);
\end{tikzpicture}
\caption{The tree form and positions of the term $f(a.b.g(\swap{a}{b}\permap X,a), h(c))$.}
\label{fig:tree}
\end{wrapfigure}

\emph{Positions} in terms are defined with respect to their tree
representation in the usual way, as strings of integers. However,
suspensions are put in a single leaf node. For instance, the
tree form of the term $f(a.b.g(\swap{a}{b}\permap X,a), h(c))$, and
the corresponding positions are shown in Fig.~\ref{fig:tree}. The
symbol $f$ stands in the position $\epsilon$ (the empty sequence). The
suspension is put in one node of the tree, at the position
$1.1.1.1$. The abstraction operator and the corresponding bound atom
together occupy one node as well. For any term $t$, $t|_\pos$ denotes
the \emph{subterm of} $t$ \emph{at position} $\pos$. For instance,
$f(a.b.g(\swap{a}{b}\permap X,a), h(c))|_{1.1}=b.g(\swap{a}{b}\permap X,a)$.

The \emph{path to a position} in a term is defined as the
sequence of expressions from the root to the node at that position (not
including) in the tree form of the term, e.g., the
path to the position $1.1.1.2$ in $f(a.b.g(\swap{a}{b}\permap X,a), h(c))$ is $f,a.,b.,g$. Since suspensions are always in leaves, they never appear in a path.

Every permutation $\ppi$ naturally defines a bijective function from
the set of atoms to the sets of atoms, that we will also represent as
$\ppi$. Suspensions are uses of variables with a permutation of atoms
waiting to be applied once the variable is instantiated. Occurrences
of an atom $a$ are said to be bound if they are in the scope of an
abstraction of $a$, otherwise are said to be free. We
  denote by $\fatom(t)$ the set of all atoms which occur freely in
  $t$: $\fatom(f(t_1,\ldots,t_n))=\bigcup_{i=1}^{n} \fatom(t_i)$,
  $\fatom(a)=\{a\}$, $\fatom(a.t)=\fatom(t)\setminus\{a\}$, and
  $\fatom(\ppi\permap X)=\atoms(\ppi)$.
  $\fatoms(t)$ is the set of all atoms which occur freely in
  $t$ ignoring suspensions:
  $\fatoms(f(t_1,\ldots,t_n))=\bigcup_{i=1}^{n} \fatoms(t_i)$,
  $\fatoms(a)=\{a\}$, $\fatoms(a.t)=\fatoms(t)\setminus\{a\}$, and
  $\fatoms(\ppi\permap X)=\emptyset$.

The head of a term $t$, denoted $\head(t)$, is defined as: $\head({ f(t_1,\ldots,t_n)})= f$, $\head({ a})= a$, $\head({
  a.t})= .$, and $\head({ \ppi\permap X}) = X$.

Substitutions are defined in the standard way, as a mapping from
variables to terms of the same sort.
We use Greek letters
$\sigma,\vartheta,\varphi$ to denote substitutions. The identity
substitution is denoted by $\ids$. 
Furthermore, we use the postfix notation for substitution applications, i.e.~$t\sigma$ denotes the application of a substitution $\sigma$ to a term $t$, and similarly, the composition of two substitutions $\sigma$ and $\vartheta$ is written as $\sigma\vartheta$. Composition of two substitutions is performed as usual.
Their application allows atom
capture, for instance, $a.X\{X\mapsto a\}= a.a$, and forces the
permutation effect: $\ppi\permap X\{X\mapsto t\}=\ppi\permef t$, for
instance, $\swap{a}{b}\permap X\{X\mapsto f(a,\swap{a}{b}\permap
Y)\}=f(b,\swap{a}{b}\swap{a}{b}\permap Y)$. The notions of substitution
\emph{domain} and \emph{range} are also standard and are denoted,
respectively, by $\dom$ and $\ran$.

A \emph{freshness constraint} is a pair of the form $a \# X$ stating
that the instantiation of $ X$ cannot contain free occurrences of
$a$. A \emph{freshness context} is a finite set of freshness
  constraints. We will use $\nabla$ and $\Gamma$ to denote freshness
  contexts. $\vars(\nabla)$ and $\atoms(\nabla)$ denote respectively
the set of variables and atoms of $\nabla$.

We say that a substitution $\sigma$ \emph{respects} a freshness
context $\nabla$, if for all $X$, $\fatoms(X\sigma)\cap
\{a \mid a\# X \in \nabla\} = \emptyset$.

The predicate $\approx$, which stands for $\alpha$-equivalence between terms, 
and the freshness predicate $ \#$ were defined in
\cite{DBLP:conf/csl/UrbanPG03,DBLP:journals/tcs/UrbanPG04} by the following theory:

\[\frac{}{ \nabla\vdash a\approx a}\text{($\approx$-atom)}\qquad \frac{ \nabla \vdash t \approx t'}{ \nabla\vdash a.t\approx a.t'}\text{($\approx$-abs-1)} \]

\[\frac{ a\neq a'\quad \nabla \vdash t \approx (a\,a')\permef t' \quad \nabla \vdash a\# t'}{ \nabla\vdash a.t\approx a'.t'}\text{($\approx$-abs-2)}\]

\[\frac{{ a\# X}\in \nabla \text{ for all $a$ such that $ \ppi \permef a \neq \ppi'\permef a$ }}{ \nabla \vdash \ppi \permap X \approx \ppi'\permap X}\text{($\approx$-susp.)}\]

\[\frac{ \nabla \vdash t_1\approx t'_1 \quad \cdots \quad \nabla \vdash t_n\approx t'_n}{ \nabla \vdash f(t_1,\ldots t_n)\approx f(t'_1,\ldots,t'_n)}(\text{$\approx$-application})\]
where the freshness predicate $ \#$ is defined by
\[\frac{ a\neq a'}{ \nabla \vdash a\# a'}\text{($\#$-atom)}\qquad \frac{ (\ppi^{-1}\permef a\# X)\in \nabla}{ \nabla \vdash a\#\ppi\permap X}\text{($\#$-susp.)}\]

\[\frac{ \nabla \vdash a\#t_1 \quad \cdots \quad \nabla \vdash a\#t_n}{ \nabla \vdash a\#f(t_1,\ldots t_n)}(\text{$\#$-application})\]

\[\frac{}{ \nabla \vdash a\# a.t}\text{($\#$-abst-1)}\qquad \frac{ a\neq a' \quad \nabla\vdash a\# t}{ \nabla \vdash a\#a'.t}\text{($\#$-abst-2)}\]

Their intended meanings are:
\begin{enumerate}
  \item $ \nabla \vdash a \# t$ holds, if for every substitution
    $\sigma$ such that $t\sigma$ is a ground term and $\sigma$ respects the freshness context $\nabla$, we have $a$ is not
    free in $ t\sigma$;
  \item $ \nabla \vdash t \approx u$ holds, if for every substitution
    $\sigma$ such that $t\sigma$ and $u\sigma$ are ground terms and $\sigma$ respects the freshness context $\nabla$, $t\sigma$ and $u\sigma$ are $\alpha$-equivalent.
\end{enumerate}

Based on the definition of the freshness predicate, we can design an
algorithm which solves the following problem:
\begin{description}
\item[Given:] A set of freshness formulas $\{ a_1\# t_1,\dots,\allowbreak
a_n\# t_n\}$.
\item[Compute:] A \emph{minimal} (with respect to $\subseteq$)
freshness context $\nabla$ such that $ \nabla\vdash
a_1\#t_1,\dots,\nabla\vdash a_n\# t_n$.
\end{description}

Such a $\nabla$ may or may not
exist, and the algorithm should detect it. 

We give a rule-based description of the algorithm, which we call $\fe$
for it is supposed to compute a freshness context. The rules operate
on pairs $F;\nabla$, where $F$ is a set of atomic freshness formulas
of the form $ a\# t$, and $\nabla$ is a freshness context. $\dotcup$
stands for disjoint union. The rules are the following:

\infrule{Del-\fe}{Delete in \fe}
    {\{{ a\# b}\}\dotcup F;\;\nabla \Lra F;\;\nabla,}
    [if $ a\neq b$.]

\infrule{Abs-\fe 1}{Abstraction in {\fe} 1}
    {\{{ a\# a.t}\}\dotcup F;\;\nabla \Lra F;\;\nabla.}

\infrule{Abs-\fe 2}{Abstraction in {\fe} 2}
    {\{{ a\# b.t}\}\dotcup F;\;\nabla \Lra \{{ a\# t}\}\cup F;\;\nabla,}
    [if $ a\neq b$.]

\infrule{Dec-\fe}{Decomposition in \fe}
    {\{{ a\# f(t_1,\ldots,t_n)}\}\dotcup F;\;\nabla \Lra  \\
     \{{ a\# t_1,\ldots,a\# t_n}\}\cup F;\;\nabla.}

\infrule{Sus-\fe}{Suspension in {\fe}}
     {\{{ a\# \ppi \permap X}\}\dotcup F;\;\nabla \Lra F;\;\{{ \ppi^{-1}\permef a\# X }\}\cup \nabla.} \vspace{0.3cm}

To compute a minimal freshness context which justifies the atomic
freshness formulas $ a_1\# t_1,\ldots, a_n\# t_n$, we start with $\{{
  a_1\# t_1,\ldots, a_n\# t_n}\};\emptyset$ and apply the rules of
${\fe}$ as long as possible. It is easy to see that the algorithm
terminates. The state to which no rule applies has either the form
$\emptyset; \nabla$ or $\{{ a\# a}\}\cup F;\nabla$, where $\nabla$ is
a freshness context. In the former case we say that the algorithm
succeeds and computes $\nabla$, writing this fact as $\fe(\{{ a_1\#
  t_1,\ldots, a_n\# t_n}\})=\nabla$. In the latter case we say that
${\fe}$ fails and write $\fe(\{{ a_1\# t_1,\ldots, a_n\#
  t_n}\})=\bot$. 

The following theorem is easy to verify:

\begin{theorem}
  \label{thm:properties:of:FC}
   Let $F$ be a set of freshness formulas and $\nabla$ be a freshness context. Then ${\fe}(F) \subseteq \nabla$ iff $\nabla\vdash a\# t$ for all ${ a\# t}\in F$.
\end{theorem}

\begin{corollary}
  \label{thm:properties:of:FC:2}
   ${\fe}(F)=\bot$ iff there is no freshness context that would justify all formulas in $F$.
\end{corollary}

Given a freshness context $\nabla$ and a substitution $\sigma$, we
define $\nabla\sigma = {\fe}(\allowbreak \{ a\# X\sigma \mid a\# X\in \nabla\})$. The following lemma is straightforward:

\begin{lemma}
  \label{lemma:fc:respects:context}
  $\sigma$ respects $\nabla$ iff $\nabla\sigma\neq \bot$.
\end{lemma}

When $\nabla\sigma\neq \bot$, we call $\nabla\sigma$ the
\emph{instance} of $\nabla$ under $\sigma$.

It is not hard to see that (a)  if $\sigma$ respects $\nabla$, then $\sigma$ respects any $\nabla'\subseteq \nabla$, and (b) if $\sigma$ respects $\nabla$ and $\vartheta$ respects $\nabla\sigma$, then $\sigma\vartheta$ respects $\nabla$ and $(\nabla\sigma)\vartheta=\nabla(\sigma\vartheta)$.

\begin{definition}
A \emph{term-in-context} is a pair $\np{\nabla}{t}$ of a freshness
context and a term. 
A term-in-context $\np{\nabla_1}{t_1}$ is \emph{more
  general} than a term-in-context $\np{\nabla_2}{t_2}$, written
$\np{\nabla_1}{t_1} \preceq \np{\nabla_2}{t_2}$, if there exists a
substitution $\sigma$, which respects $\nabla_1$, such that $
\nabla_1\sigma \subseteq \nabla_2$ and $ \nabla_2\vdash t_1\sigma
\approx t_2$. 

We write $\nabla \vdash t_1 \preceq t_2$ if there exists a substitution $\sigma$ such that $\nabla \vdash t_1\sigma \approx t_2$.

Two terms-in-context $p_1$ and $p_2$ are \emph{equivalent} (or
\emph{equi-general}), written $p_1\simeq p_2$, iff $p_1 \preceq p_2$
and $p_2 \preceq p_1$. The strict part of $\preceq$ is denoted by $\prec$, i.e.,
$p_1\prec p_2$ iff $p_1 \preceq p_2$ and not $p_2 \preceq p_1$.
We also write $\nabla \vdash t_1 \simeq t_2$ iff $\nabla \vdash t_1 \preceq t_2$ and $\nabla \vdash t_2 \preceq t_1$.
\end{definition}

\begin{example}
We give some examples to demonstrate the relations we have just defined:
\begin{itemize}
\item $\np{\{a\#X\}}{f(a)} \simeq \np{\emptyset}{f(a)}$. We can use $
  \{X\mapsto b\}$ for the substitution applied to the first pair.
\item $\np{\emptyset}{f(X)} \preceq \np{\{a\#X\}}{f(X)}$ (with
  $\sigma=\ids$), but not $\np{\{a\#X\}}{f(X)} \preceq
  \np{\emptyset}{f(X)}$.
\item $\np{\emptyset}{f(X)}\preceq \np{\{a\#Y\}}{f(Y)}$ with $\sigma= \{X\mapsto Y\}$.
\item $\np{\{a\#X\}}{f(X)}\not \preceq \np{\emptyset}{f(Y)}$, because
  in order to satisfy $ \{a\# X\}\sigma \subseteq \emptyset$, the
  substitution $\sigma$ should map $ X$ to a term $ t$ which contains
  neither $a$ (freely) nor variables. But then $ \emptyset \vdash f(t)
  \approx f(Y)$ does not hold.  Hence, together with the previous
  example, we get $\np{\emptyset}{f(Y)}\prec \np{\{a\#X\}}{f(X)}$.
\item $\np{\{a\#X\}}{f(X)} \not\preceq \np{\{a\#X\}}{f(a)}$.  Notice
  that $\sigma=\{{ X\mapsto a}\}$ does not respect $\{{ a\#X}\}$.
\item $\np{\{b\#X\}}{\swap{a}{b}\permap X} \preceq \np{\{c\#X\}}{\swap{a}{c}\permap
  X} $ with the substitution $\sigma=\{ {X\mapsto} \allowbreak \swap{a}{b}\swap{a}{c}\permap
  X\}$.\, Hence,\, we\, get\, $\np{\{b\#X\}}{\swap{a}{b}\permap X}\, \simeq\,
  \np{\{c\#X\}}{\swap{a}{c}\permap X} $,\, because\, the\, $\succeq$\, part\, can\, be\,
  shown\, with\, the\, help\, of\, the\, substitution\, $\{ X\mapsto\allowbreak
    \swap{a}{c}\swap{a}{b}\permap X\}$.
\end{itemize}
\end{example}

\begin{definition}
A term-in-context $\np{\Gamma}{r}$ is called a \emph{generalization} of two terms-in-context $\np{\nabla_1}{t}$ and
$\np{\nabla_2}{s}$ if $\np{\Gamma}{r} \preceq \np{\nabla_1}{t}$
and $\np{\Gamma}{r} \preceq \np{\nabla_2}{s}$. It is the \emph{least
  general generalization,} (lgg in short) of $\np{\nabla_1}{t}$ and $\np{\nabla_2}{s}$ if
there is no generalization $\np{\Gamma'}{r'}$ of $\np{\nabla_1}{t}$
and $\np{\nabla_2}{s}$ which satisfies $\np{\Gamma}{r} \prec
\np{\Gamma'}{r'}$.
\end{definition}

Note that if we have infinite number of atoms in the language, the
relation $\prec$ is not well-founded: $\np{\emptyset}{X}\prec
\np{\{a\#X\}}{X}\prec \np{\{a\#X,b\#X\}}{X}\prec \cdots$. As a
consequence, two terms-in-context may not have an lgg and not even a
minimal complete set of generalizations:\footnote{Minimal complete
  sets of generalizations are defined in the standard way. For a
  precise definition, see, e.g.,
  \cite{DBLP:conf/lopstr/AlpuenteEMO08,DBLP:conf/rta/KutsiaLV11}.}

\begin{example}\label{exmp:nullary}
  Let $p_1=\np{\emptyset}{a_1}$ and $p_2=\np{\emptyset}{a_2}$ be two
  terms-in-context. Then in any complete set of generalizations of $p_1$
  and $p_2$ there is an infinite chain $\np{\emptyset}{X}\prec
  \np{\{a_3\#X\}}{X}\prec \np{\{a_3\#X,a_4\#X\}}{X}\prec \cdots$,
  where $\{a_1,a_2,a_3,\ldots\}$ is the set of all atoms
  of the language. Hence, $p_1$ and $p_2$ do not have a minimal
  complete set of generalizations.
\end{example} 

This example is a proof of the theorem, which characterizes the
generalization type of nominal
anti-unification:\footnote{Generalization types are defined
  analogously to unification types: unary, finitary,
    infinitary, and nullary, see~\cite{DBLP:conf/rta/KutsiaLV11}.}

\begin{theorem}
  \label{thm:generalization:type}
   The problem of anti-unification for terms-in-context is of nullary type.
\end{theorem}

However, if we restrict the set of atoms which can be used in the
generalizations to be fixed and finite, then the anti-unification problem
becomes unitary. (We do not prove this property here, it will follow
from the Theorems~\ref{thm:completeness} and~\ref{thm:uniqueness} in
Sect.~\ref{sect:properties}.)

\begin{definition}
We say that a term $ t$ (resp., a freshness context $\nabla$)
is \emph{based} on a set of atoms $A$ iff 
$\atoms({ t})\subseteq A$ (resp., $\atoms(\nabla)\subseteq
A$). A term-in-context $\np{\nabla}{t}$ is based on~$A$ if both
$ t$ and $\nabla$ are based on it. We extend the notion of $A$-basedness to permutations, calling $\ppi$ $A$-based if it contains only atoms from $A$. 
Such a permutation defines a bijection, in particular, from $A$ to $A$. 
If $p_1$ and $p_2$ are $A$-based
terms-in-context, then their \emph{$A$-based generalizations} are
terms-in-context which are generalizations of $p_1$ and $p_2$ and are
based on~$A$. An \emph{$A$-based lgg} of $A$-based
terms-in-context $p_1$ and $p_2$ is a term-in-context $p$, which is an
$A$-based generalization of $p_1$ and $p_2$ and there is no
$A$-based generalization $p'$ of $p_1$ and~$p_2$ which satisfies
$p\prec p'$. 
\end{definition}

The problem we would like to solve is the following: 
\vspace{0.2cm}
\begin{description}
  \item[\bf Given:] Two nominal terms $ t$ and $ s$ of the same
    sort, a freshness context $\nabla$, and a \emph{finite} set of atoms $A$
    such that $ t$, $ s$, and $\nabla$ are based on $A$.  
  \item[\bf Find:] A term $ r$ and a freshness context $\Gamma$, such
    that the term-in-context $\np{\Gamma}{r}$ is an $A$-based least
    general generalization of the terms-in-context $\np{\nabla}{t}$
    and $\np{\nabla}{s}$.
\end{description}
\vspace{0.2cm}

Our anti-unification problem is parametric on the set of atoms we
consider as the base, and finiteness of this set is essential to
ensure the existence of an lgg.  

\section{Motivation of Using a Direct Nominal Anti-Unification Algorithm}
\label{sect:motivation}

In~\cite{DBLP:journals/tocl/LevyV12}, relation between nominal
unification (NU) and higher-order pattern unification (HOPU) has been
studied. In particular, it was shown how to translate NU problems into
HOPU problems and how to obtain nominal unifiers back from
higher-order pattern unifiers. It is tempting to use the same
translation for nominal anti-unification (NAU), using the algorithm
from~\cite{DBLP:conf/rta/BaumgartnerKLV13} to solve higher-order
anti-unification problems over patterns (HOPAU), but it turns out that
the generalization computed in this way is not always based on the
same set of the atoms as the input:

\begin{example}
\label{exmp:translation}
We consider the following problem:
Let the set of atoms be $A_1=\{a,b\}$. The terms to be generalized are $a.b$ and $b.a$, and the freshness context is $\nabla=\emptyset$.
According to~\cite{DBLP:journals/tocl/LevyV12}, translation to higher-order patterns
gives the anti-unification problem $\lambda a,b,a.\,b \triangleq
\lambda a,b,b.\,a$, whose lgg is $\lambda a,b,c. X(a,b)$. However, we
can not translate this lgg back to an $A_1$-based term-in-context,
because it contains more bound variables than there are atoms in
$A_1$.

On the other hand, the translation would work for the 
set of atoms $A_2=\{a,b,c\}$:
Back-translating $\lambda
a,b,c. X(a,b)$ gives the $A_2$-based lgg $\np{\{c\#X\}}{c.X}$.
\end{example}

The reason why the translation-based approach does not work for
$A$-based NAU is that $A$ is finite, while in higher-order
anti-unification there is an infinite supply of fresh (bound)
variables. 
If we assumed $A$ to be infinite, there would still be a mismatch between
NAU and the corresponding HOPAU: NAU, as we saw, is nullary in this
case, while HOPAU is unitary. The reason of this contrast is that from
infinitely many nominal generalizations, there is only one which is a
well-typed higher-order generalization.

One might think that the translation-based approach would still work,
if one considers only nominal anti-unification problems where the set
of atoms is large enough for the input
terms-in-context.  However, there
is a reason that speaks against NAU-to-HOPAU translation:
complexity. The translation approach leads to a quadratic increase of
the input size (Lemma 5.6 in \cite{DBLP:journals/tocl/LevyV12}). The
HOPAU algorithm in \cite{DBLP:conf/rta/BaumgartnerKLV13} runs in cubic
time  with respect to the size of its
input. Hence, the translation-based approach leads to an algorithm
with runtime complexity $O(n^6)$. In
contrast, the algorithm developed in this paper has runtime complexity
$O(n^5)$, and requires no back and forth
translations.

\section{The Lattice of More General Terms-In-Context}
\label{sect:lattice}

The notion of \emph{more general term} defines an order relation
between classes of terms (modulo some notion of \emph{variable
  renaming}). In most cases, we have actually a meet-semilattice,
since, given two terms, there always exists a greatest lower bound
(meet) that corresponds to their anti-unifier. On the contrary, the
least upper bound (join) of two terms only exists if they are
unifiable. For instance, the two first-order terms $f(a,X_1)$ and
$f(X_2,b)$ have a meet $f(Y_1,Y_2)$, and, since they are unifiable,
also a join $f(a,b)$. Notice that unifiability and existence of a join
are equivalent if both terms do not share variables (for instance
$f(a,X)$ and $f(X,b)$ are both smaller than $f(a,b)$, hence joinable,
but they are not unifiable). With this restriction one do not loose
generality: The unification problem $t_1\approx^? t_2$ (sharing variables),
can be reduced to $f(t_1,t_2) \approx^?  f(X,X)$ (not sharing variables), where $f$
is some binary symbol and $X$ a fresh variable.  Therefore, in the
first-order case, the problem of searching a most general unifier is
equivalent to the search of the join of two terms, and the search of a
least general generalization 
to the search of the meet. Notice that meet
and join are unique up to some notion of \emph{variable renaming}.
For instance, the join of $f(a,X,X')$ and $f(Y,b,Y')$ is $f(a,b,Z)$
for any renaming of $Z$ by any variable.

In the nominal case, we consider the set of terms-in-context (modulo variable
renaming) with the more general relation. The following lemma
establishes a correspondence between joinability and unifiability.

\begin{restatable}{lemma}{joinable}
 \label{lemma:unif:join}
Given two terms-in-context $\np{\nabla_1}{t_1}$ and
$\np{\nabla_2}{t_2}$ with disjoint sets of variables,
$\np{\nabla_1}{t_1}$ and $\np{\nabla_2}{t_2}$ are joinable if, and
only if, $\{t_1\approx^? t_2\}\cup\nabla_1\cup\nabla_2$ has a solution
(is unifiable).
\end{restatable}

\begin{proof}
If they are joinable, there exist $\np{\Gamma}{s}$ such that
$\np{\nabla_i}{t_i} \preceq \np{\Gamma}{s}$. Hence, there exist
$\sigma_1$ and $\sigma_2$ such that $\sigma_i$ respects $\nabla_i$ and
$\nabla_i\sigma_i\subseteq\Gamma$ and $\Gamma\vdash t_i\sigma_i
\approx s$.  Now, since $\vars(t_1) \cap \vars(t_2)=\emptyset$, define
$\sigma(X) = \sigma_1(X)$, if $X\in\vars(t_1)$, and $\sigma(X) =
\sigma_2(X)$, if $X\in\vars(t_2)$. We have $\Gamma\vdash
t_1\sigma\approx s\approx t_2\sigma$ and
$(\nabla_1\cup\nabla_2)\sigma\subseteq\Gamma$.  Hence, according to
the properties of $\fe$ as stated in Theorem~\ref{thm:properties:of:FC}, we have $\Gamma\vdash a\#X\sigma$, for any $a\#X\in
\nabla_1\cup\nabla_2$.  Therefore, the pair
$\langle\Gamma,\sigma\rangle$ is a nominal unifier of $\{t_1\approx^?
t_2\}\cup\nabla_1\cup\nabla_2$.

Conversely, we can prove easily that if $\langle\Gamma,\sigma\rangle$
solves $\{t_1\approx^?  t_2\}\cup\nabla_1\cup\nabla_2$, then
$\np{\nabla_i}{t_i} \preceq \np{\Gamma}{s}$, where $s=t_1\sigma$.
Let's prove now that when $\langle\Gamma,\sigma\rangle$ is the most
general nominal unifier, then $\np{\Gamma}{s}$ is the join of
$\np{\nabla_1}{t_1}$ and $\np{\nabla_2}{t_2}$, i.e.  whenever
$\np{\nabla_i}{t_i}\preceq\np{\Gamma'}{s'}$, we have
$\np{\Gamma}{s}\preceq\np{\Gamma'}{s'}$:

From $\np{\nabla_i}{t_i}\preceq\np{\Gamma'}{s'}$ we have that there exists
$\sigma'$ such that $\langle\Gamma',\sigma'\rangle$ is a nominal
unifier of $\{t_1\approx^?  t_2\}\cup\nabla_1\cup\nabla_2$ and
$\Gamma'\vdash t_1\sigma'\approx s' \approx t_2\sigma'$. Since $\langle\Gamma,\sigma\rangle$ is most general, there exists a
substitution $\varphi$, which respects $\Gamma$, such that $\Gamma\varphi\subseteq\Gamma'$ and
$\Gamma'\vdash t_i\sigma' \approx t_j\sigma\varphi $, for all $i,j\in\{1,2\}$.  The existence of this
substitution $\varphi$ proves
$\np{\Gamma}{s}\preceq\np{\Gamma'}{s'}$.
\end{proof}

Like\, in\, first-order\, unification,\, the\, previous\, lemma\, allows\, us\, to\,
reduce any\, nominal\, unification\, problem\,
$P\,=\,\{a_1\#u_1,\,\dots,\,a_m\#u_m,\;t_1\approx s_1,\,\dots,$ $t_n\approx s_n\}$
into the joinability of the two terms-in-context
$\np{\emptyset}{f(X,X)}$ and
$\np{\fe(\allowbreak\{a_1\#u_1,\dots,\allowbreak a_m\#u_m\})}{f(g(t_1,\dots,t_n),g(s_1,\dots, \allowbreak s_n))}$
where $f$ and $g$ are any appropriate function symbols, and $X$ is a
fresh variable.

The nominal anti-unification problem is already stated in terms of
finding the meet of two terms-in-context, with the only proviso that
all terms and contexts must be based on some finite set of atoms.

\section{Nominal Anti-Unification Algorithm}
\label{sect:nau:alg}

The triple $X: t \triangleq s$, where $X,t,s$ have the same sort, is
called the \emph{anti-unification triple}, shortly AUT, and the
variable $X$ is called a \emph{generalization variable}. We say that
a set of AUTs $P$ is based on a finite set of atoms $A$, if for all ${
  X: t \triangleq s}\in P$, the terms $t$ and $s$ are based on $A$.

\begin{definition}
The nominal anti-unification algorithm is formulated in a rule-based
way working on tuples 
$P;\, S;\, \nablag;\,\sigma$ and two global parameters
  $\Atoms$ and $\nablai$, where
\begin{itemize}
 \item $P$ and $S$ are sets of AUTs such that if ${ X: t \triangleq s}
   \in P\cup S$, then this is the sole occurrence of $ X$ in $ P\cup
   S$;
 \item $P$ is the set of AUTs to be solved;
 \item $\Atoms$ is a finite set of atoms;
 \item The freshness context $ \nablai$ does not constrain generalization variables;
 \item $S$ is a set of already solved AUTs (the store);
 \item $\nablag$ is a freshness context (computed so far) which
   constrains generalization variables;
 \item $\sigma$ is a substitution (computed so far) mapping
   generalization variables to nominal terms;
 \item $P$, $S$, $\nablai$, and $\nablag$ are $\Atoms$-based.
\end{itemize}

We call such a tuple a \emph{state}. The rules below operate on
states.

\end{definition}

\infrule{Dec}{Decomposition}
    {\{X:  \nomh(t_1,\ldots,t_m) \triangleq \nomh(s_1,\ldots,s_m)\}\dotcup P;\; S;\; \nablag;\; \sigma \\ \Lra   
     \{Y_1:  t_1\triangleq s_1,\ldots, Y_m : t_m\triangleq s_m\}\cup P;\; S; \; \nablag ;\; 
      \sigma \{X\mapsto \nomh(Y_1,\ldots,Y_m)\}, }
[\noindent where $\nomh$ is a function symbol or an atom, 
  ${ Y_1,\ldots,Y_m}$ are fresh variables of the corresponding sorts, $m\ge 0$.]

\infrule{Abs}{Abstraction}
    {\{ X :  a.t \triangleq b.s \}\dotcup P;\; S ;\; \nablag;\;  \sigma \Lra   
     \{ Y: \swap{c}{a}\permef t \triangleq \swap{c}{b}\permef s \}\cup P; \; S ;\; 
    \nablag;\;  
     \sigma\{X\mapsto  c.Y\}, }
[\noindent where $ Y$ is fresh, $ c\in \Atoms$, $ \nablai \vdash c\# a.t$ and $ \nablai \vdash c\# b.s$.
]

\infrule{Sol}{Solving}
  {\{{ X:  t \triangleq s}\}\dotcup P;\; S;\; \nablag;\; \sigma \Lra 
     P;\;  S\cup\{{ X:  t \triangleq s}\};\;\Gamma \cup \Gamma';\; \sigma, }
[]
\noindent if none of the previous rules is applicable, 
i.e. one of the following conditions hold:
\begin{itemize}
\item[(a)] both terms have distinct heads: ${ \head(
  t)}\neq { \head(s)}$, or
\item[(b)] both terms are suspensions: $ t= \ppi_1 \permap
  Y_1$ and $ s=\ppi_2\permap Y_2$, where $\ppi_1,\ppi_2$ and $ Y_1,Y_2$ are
  not necessarily distinct, or
\item[(c)] both are abstractions and rule \textsf{Abs} is not
  applicable: $ t=a.t'$, $ s=b.s'$ and there is no atom $c \in \Atoms$
  satisfying $\nablai \vdash c\# a.t'$ and $\nablai \vdash c\#
  b.s'$.
\end{itemize}
The set $\Gamma'$ is defined as $ \Gamma':= \{ a\#X \mid a\in \Atoms\ \wedge\ 
  \nablai \vdash a\#t\ \wedge\ \nablai \vdash a\#s\}.$

\infrule{Mer}{Merging} {P;\; \{{ X : t_1 \triangleq s_1, Y : t_2
    \triangleq s_2}\}\dotcup S ;\; \nablag;\; \sigma \Lra \\
    P; \; \{{ X : t_1 \triangleq s_1}\}\cup S ;\; 
     \nablag{\{Y \mapsto \ppi\permap X\}};\; { \sigma\{ Y \mapsto \ppi\permap
    X\}}, } [\noindent where  $\ppi$ is an $\atoms(t_1,s_1,t_2,s_2)$-based
  permutation such that $\nablai\vdash \ppi\permef t_1\approx t_2$,
  and $\nablai\vdash \ppi\permef s_1\approx s_2$.]\vspace{0.3cm}

The rules transform states to states. One can easily observe this by inspecting the rules. 

Given a finite set of atoms $\Atoms$, two nominal $\Atoms$-based terms
$ t$ and $ s$, and an $\Atoms$-based freshness context $\nablai$, to
compute $\Atoms$-based generalizations for
$\np{\nablai}{t}$ and $\np{\nablai}{s}$, we start with $\{{ X : t
  \triangleq s}\}; \emptyset; \emptyset; \ids$, where $ X$ is a fresh
variable, and apply the rules as long as possible.  We denote this
procedure by $\frN$. A \emph{Derivation} is a sequence of state transformations by the rules. The state to which no rule applies has the form
$\emptyset; S; \nablag; \varphi$, where \textsf{Mer} does not apply to
$S$. We call it the \emph{final state.} When $\frN$ transforms $\{ X : t
\triangleq\penalty 10000 s\}; \emptyset; \emptyset; \ids$ into a final
state $\emptyset; S; { \nablag;} \varphi$, we say that the
\emph{result computed} by $\frN$ is $\np{\nablag}{X\varphi}$.

Note that the {\sf Dec} rule works also for the AUTs of the form
$X:a\triangleq a$. In the {\sf Abs} rule, it is important to have the
corresponding $c$ in $\Atoms$. If we take $\Atoms=A_2$ in
Example~\ref{exmp:translation}, then {\sf Abs} can transform the AUT
between $t$ and $s$ there, but if $\Atoms=A_1$ in the same example,
then {\sf Abs} is not applicable. In this case the {\sf Sol} rule
takes over, because the condition (c) of this rule is satisfied.

The condition (b) of {\sf Sol} helps to compute, e,g,
$\np{\emptyset}{X}$ for identical terms-in-context
$\np{\emptyset}{\swap{a}{b}\permap Y}$ and $\np{\emptyset}{\swap{a}{b}\permap
  Y}$. Although one might expect that computing
$\np{\emptyset}{\swap{a}{b}\permap Y}$ would be more natural, from the
generalization point of view it does not matter, because
$\np{\emptyset}{X}$ is as general as $\np{\emptyset}{\swap{a}{b}\permap Y}$.

\begin{example}\label{exmp:nau}
We illustrate $\frN$ with the help of some examples:
\begin{itemize}
\item Let $ t=f(a,b)$, $ s=f(b,c)$, $ \nablai =\emptyset$, and
$\Atoms=\{a,b,c,d\}$. Then $\frN$ performs the following
transformations:
  \begin{alignat*}{3}
  &\{ X: f(a,b) \triangleq f(b,c) \};  \emptyset;\emptyset; \ids \Lra_{\mathsf{Dec}}{} \\
  & \{ { Y: a \triangleq b, Z: b \triangleq  c }\};   \emptyset;\emptyset;\{ { X \mapsto f(Y,Z )}\} \Lra^2_{\mathsf{Sol}}{} \\
  & \emptyset;  \{ { Y: a \triangleq  b, Z: b \triangleq  c }\}; \{{ c\# Y, d\#Y, a\# Z, d\# Z}\};\; \{ X \mapsto f(Y,Z)\} \Lra_{\mathsf{Mer}}{} \\
  & \emptyset; \{ { Y: a \triangleq b}\};  \{{ c\# Y, d\# Y}\};\{ { X \mapsto f(Y,\swap{a}{b}\swap{b}{c}\permap Y)}\}
\end{alignat*}
Hence, $p=\np{\{c\# Y, d\#Y \}}{f(Y,\swap{a}{b}\swap{b}{c}\permap Y)}$ is the
computed result. It generalizes the input pairs: $p\{ Y \mapsto a\} \preceq
\np{\nablai}{t} $ and $p\{Y \mapsto
  b\} \preceq \np{\nablai}{s}$. The substitutions $\{ Y \mapsto a\}$ and $\{Y \mapsto
  b\}$ can be read from
the final store. Note that $\np{\{c\# Y\}}{f(Y,\swap{a}{b}\swap{b}{c}\permap
  Y)}$ would be also an $\Atoms$-based generalization of
$\np{\nablai}{t}$ and $\np{\nablai}{s}$, but it is strictly more
general than~$p$.

\item
Let $ t=f(b,a)$, $ s=f(Y,\swap{a}{b}\permap Y)$, $ \nablai =\{b\# Y\}$, and
$\Atoms=\{a,b\}$. Then $\frN$ computes the term-in-context
$\np{\emptyset}{f(Z,\swap{a}{b}\permap Z)}$. It generalizes the input
pairs.

\item
Let $ t=f(g(X),X)$, $ s=f(g(Y),Y)$, $ \nablai =\emptyset$, and
$\Atoms=\emptyset$. It is a first-order anti-unification problem. $\frN$ computes
$\np{\emptyset}{f(g(Z),Z)}$. It generalizes the input pairs.

\item
Let $ t = f(a.b,X)$, $ s= f(b.a,Y)$, $ \nablai =\{c\# X\}$,
$\Atoms=\{a,b,c,d\}$. Then $\frN$ computes the term-in-context
$p=\np{\{c\# Z_1,d\#Z_1\}}{f(c.Z_1,Z_2)}$. It generalizes the input
pairs: $p\{Z_1\mapsto b,Z_2\mapsto X\}=\np{\emptyset}{f(c.b,X)}\preceq \np{\nablai}{t}$ and
  $p\{Z_1\mapsto a,Z_2\mapsto Y\} =\np{\emptyset}{f(c.a,Y)}\preceq \np{\nablai}{s}$.  

\end{itemize}
\end{example}

\section{Properties of the Nominal Anti-Unification Algorithm}
\label{sect:properties}

The Soundness Theorem states that the result computed by $\frN$ is indeed an $A$-based generalization of the input terms-in-context:
\begin{restatable}[Soundness of $\frN$]{theorem}{soundnessN}
  \label{thm:soundness}
  Given terms $t$ and $s$ and a freshness context $\nablai$,
  all based on a finite set of atoms $\Atoms$, if  $\{X: t\triangleq s\};\, \emptyset;\, \emptyset;\, $
  $\varepsilon \Lra^{+} \emptyset;\,S;$ $ \nablag;\, \sigma$ is a derivation obtained by an execution of $\frN$,
  then $\np{\nablag}{X\sigma}$ is an $\Atoms$-based generalization of
  $\np{\nablai}{t}$ and $\np{\nablai}{s}$.
\end{restatable}

\begin{proof} Since all atoms introduced by the rules
    of $\frN$ in $\nablag$ and $\sigma$ are from $\Atoms$,
    $\np{\nablag}{X\sigma}$ is $\Atoms$-based. To prove that
    $\np{\nablag}{X\sigma}$ generalizes both $\np{\nablai}{t}$ and
    $\np{\nablai}{s}$, we use well-founded induction on the length of
    derivations. In fact, we will prove a more general statement:

    Assume $P_0;S_0;\Gamma_0;\vartheta_0 \Lra^+ \emptyset;
    S_n;\Gamma_n;\vartheta_0\vartheta_1\cdots\vartheta_n$ is a
    derivation in $\frN$ (with $\nabla$ and $\Atoms$) with
      the following property: If $Z_0:t_0\triangleq s_0\in S_0$, then
      $a\# Z_0\in \nablag_0$ for an $a\in \Atoms$ iff $\nabla\vdash
      a\#t_0$ and $\nabla\vdash a\#s_0$. Notice that
      requiring this property does not imply a lose of generality: our
      algorithm starts with no equation in the store, and each time an
      equation is moved to the store the {\sf Sol} rule adds the
      required freshness constraints (by inspection of {\sf
        Sol}). Moreover, freshness constraints are only removed from
      the freshness context when {\sf Mer} removes the corresponding
      equation from the store (by inspection of {\sf Mer}).  Then for
    any $Z_0:t_0 \triangleq s_0\in P_0\cup S_0$ we have
    $\np{\Gamma_n\setminus \Gamma_0}{Z_0 \vartheta_1\cdots\vartheta_n}
    \preceq \np{\nabla}{t_0}$ and $\np{\Gamma_n\setminus \Gamma_0}{Z_0
      \vartheta_1\cdots\vartheta_n} \preceq \np{\nabla}{s_0}$.

Assume the statement is true for any derivation of the length $l < n$
and prove it for a derivation $P_0;S_0;\Gamma_0;\vartheta_0 \Lra^+
\emptyset; S_n;\Gamma_n;\vartheta_0\vartheta_1\cdots\vartheta_n$ of
the length $n$. 

Below the composition $\vartheta_i\vartheta_{i+1}\cdots\vartheta_k$ is abbreviated as $\vartheta_i^k$ with $ k\ge i$. 

Let $Z_0: t_0 \triangleq s_0$ be an AUT selected for transformation from $P_0\cup S_0$.  We consider each rule:\vspace{0.3cm}

{\sf Dec:} $Z_0=X$, $t_0=\nomh(t_1,\ldots,t_m)$,
$s_0=\nomh(s_1,\ldots,s_m)$, $\Gamma_1=\Gamma_0$, and
$\vartheta_1=\{X\mapsto \nomh(Y_1,\ldots,Y_m)\}$. By the induction
hypothesis (IH), $\np{\Gamma_n\setminus
  \Gamma_1}{Y_i\vartheta_2^n}\preceq \np{\nabla}{t_i}$ and
$\np{\Gamma_n\setminus \Gamma_1}{Y_i\vartheta_2^n}\preceq
\np{\nabla}{s_i}$ for all $1\le i \le m$.  
Hence by definition of $\preceq$ for terms-in-context,
  there exist substitutions $\sigma$ and $\varphi$ such that:
\begin{itemize}
\item $\vartheta_2^n\sigma$  and $\vartheta_2^n\varphi$ respect
  $\Gamma_n\setminus \Gamma_1$,
\item  $(\Gamma_n\setminus
  \Gamma_1)\sigma\subseteq\nabla$ and $(\Gamma_n\setminus
  \Gamma_1)\varphi\subseteq\nabla$, and
\item  $\nabla \vdash
  Y_i\vartheta_2^n\sigma \approx t_i$ and $\nabla \vdash
  Y_i\vartheta_2^n\varphi \approx s_i$ for all $1\le i \le m$.
\end{itemize}
 Finally, since $\vartheta_1=\{X\mapsto \nomh(Y_1,\ldots,Y_m)\}$ and 
$\Gamma_n\setminus \Gamma_0=\Gamma_n\setminus \Gamma_1$, we obtain
$\np{\Gamma_n\setminus \Gamma_0}{Z_0\vartheta_1^n}\preceq
\np{\nabla}{t_0}$ and $\np{\Gamma_n\setminus
  \Gamma_0}{Z_0\vartheta_1^n}\preceq \np{\nabla}{s_0}$. \vspace{0.3cm}

{\sf Abs:} $Z_0=X$, $t_0= a.t$, $s_0=b.s$, $\Gamma_1=\Gamma_0$, and
$\vartheta_1=\{X\mapsto c.Y\}$, where $ \nabla \vdash c\# a.t$ and $
\nabla \vdash c\# b.s$. $P_1$ contains the AUT $Y: \swap{c}{a}\permef t
\triangleq \swap{c}{b}\permef s$. By the IH, 
$\np{\Gamma_n\setminus
  \Gamma_1}{Y_i\vartheta_2^n}\preceq \np{\nabla}{\swap{c}{a}\permef t}$ and
$\np{\Gamma_n\setminus \Gamma_1}{Y_i\vartheta_2^n}\preceq
\np{\nabla}{\swap{c}{b}\permef s}$
 hence $\nabla \vdash
Y\vartheta_2^n\sigma \approx \swap{c}{a}\permef t$ and $\nabla \vdash
Y\vartheta_2^n\varphi \approx \swap{c}{b}\permef s$ for some $\sigma$ and
$\varphi$ that in addition satisfy the other properties (as above) for $\preceq$. Then, since we also have that $ \nabla \vdash c\# a.t$ and $
\nabla \vdash c\# b.s$ we can prove, with the $\approx$-abs rules, that $\nabla \vdash c.Y\vartheta_2^n\sigma \approx a.t$ and $\nabla
\vdash c.Y\vartheta_2^n\varphi \approx b.s$.
 Finally, since $\vartheta_1=\{X\mapsto c.Y\}$ and
$\Gamma_n\setminus \Gamma_0=\Gamma_n\setminus \Gamma_1$, we obtain
$\np{\Gamma_n\setminus \Gamma_0}{Z_0\vartheta_1^n}\preceq
\np{\nabla}{t_0}$ and $\np{\Gamma_n\setminus
  \Gamma_0}{Z_0\vartheta_1^n}\preceq \np{\nabla}{s_0}$.\vspace{0.3cm}

{\sf Sol:} $Z_0=X$, $t_0=t$, $s_0=s$, $\Gamma_1\setminus
\Gamma_0=\{a\#X \mid a\in \Atoms, \nabla \vdash a\#t,$ and $\nabla
\vdash a\#s \}$ and $\vartheta_1=\varepsilon$.  By the IH we have that
$(\Gamma_n\setminus \Gamma_1)\sigma \subseteq \nabla$,
$(\Gamma_n\setminus \Gamma_1)\varphi \subseteq \nabla$, $\nabla \vdash
X\vartheta_2^n\sigma \approx t$, and $\nabla \vdash
X\vartheta_2^n\varphi \approx s$ for some $\sigma$ and $\varphi$
respecting $\Gamma_n\setminus \Gamma_1$.

Since $\vartheta_1=\varepsilon$, from the IH we get $\nabla \vdash
X\vartheta_1^n\sigma \approx t$. To show that $(\Gamma_n\setminus
\Gamma_0)\sigma \subseteq \nabla$ take $a\# Y\in \Gamma_n\setminus
\Gamma_0$ for some $a$.

\begin{itemize}
\item  If $a\# Y\in\Gamma_n\setminus \Gamma_1$, then $\{a\# Y\}\sigma \subseteq \nabla$ by the IH,
\item otherwise, if $a\# Y\notin \Gamma_n\setminus
\Gamma_1$, then  $a\# Y\in \Gamma_1\setminus \Gamma_0$
with $X=Y$ and $X\vartheta_2^n =X$. By the IH, $\nabla
\vdash X\sigma \approx t$, besides, we know $\nabla \vdash
a\#t$. Therefore, we know $\nabla \vdash a\#X\sigma$, which by
Theorem~\ref{thm:properties:of:FC} implies $\{a\# X\}\sigma=\{a\#
Y\}\sigma \subseteq \nabla$. Thus, $(\Gamma_n\setminus \Gamma_0)\sigma
\subseteq \nabla$.
\end{itemize}

 Hence, we proved $\np{\Gamma_n\setminus
   \Gamma_0}{X\vartheta_1^n}\preceq \np{\nabla}{t}$, which is the same
 as $\np{\Gamma_n\setminus \Gamma_0}{Z_0\vartheta_1^n}\preceq
 \np{\nabla}{t_0}$.  $\np{\Gamma_n\setminus
   \Gamma_0}{Z_0\vartheta_1^n}\preceq \np{\nabla}{s_0}$ can be proved
 analogously. \vspace{0.3cm}

{\sf Mer:} First, we show that the following holds : For all $k\ge 0$, If $Z_k:t_k\triangleq s_k \in S_k$ and $c\#Z_k\in \Gamma_k$ for a $c\in A$, then $\nabla \vdash c\# t_k$ and $\nabla \vdash c\# s_k$. 

Proceed by induction on $k$. If $k=0$, then it follows from the
assumption on $P_0;S_0,\Gamma_0;\vartheta_0$. Assume it is true for
$k$ and show it for $k+1$. Take $Z_{k+1}:t_{k+1}\triangleq s_{k+1} \in
S_{k+1}$. We have two alternatives: Either $Z_{k+1}:t_{k+1}\triangleq
s_{k+1}$ has been a subject of the {\sf Mer} rule at this step, or
not. If not, then either it was already in $S_k$ or was introduced at
this step. In either case, by IH or because it has been
  introduced with {\sf Sol} rule, if $c\#Z_k\in \Gamma_k$ for a $c\in
A$, then $\nabla \vdash c\# t_k$ and $\nabla \vdash c\# s_k$. If
$Z_{k+1}:t_{k+1}\triangleq s_{k+1}$ was a subject of the {\sf Mer}
rule, then there exists some $U_k:r_k \triangleq q_k \in S_k$, such
that $\nabla \vdash \ppi_k\permef t_{k+1} \approx r_k$, $\nabla \vdash
\ppi_k\permef s_{k+1} \approx q_k$. Moreover, for all $d\# U_k \in S_k$ we now have
$\ppi_k^{-1}\permef d \# Z_{k+1} \in S_{k+1}$, and all $c\# Z_{k+1}
\in S_k$ are retained in $S_{k+1}$. For these $c$'s, since
$Z_{k+1}:t_{k+1}\triangleq s_{k+1} \in S_k$, by the induction
hypothesis we have $\nabla \vdash c\# t_{k+1}$ and $\nabla \vdash c\#
s_{k+1}$. As for $\ppi_k^{-1}\permef d \# Z_{k+1} \in S_{k+1}$, here
we need to show $\nabla \vdash \ppi_k^{-1}\permef d \# t_{k+1}$ and
$\nabla \vdash \ppi_k^{-1}\permef d \# s_{k+1}$. By the induction
hypothesis we know $\nabla\vdash d\# r_k$. Then $\nabla\vdash
\ppi_k^{-1}\permef d\# \ppi_k^{-1}\permef r_k$ and since $\nabla
\vdash \ppi_k\permef t_{k+1} \approx r_k$, we get $\nabla \vdash
\ppi_k^{-1}\permef d \# t_{k+1}$. $\nabla \vdash \ppi_k^{-1}\permef d
\# s_{k+1}$ can be shown similarly, using $\nabla\vdash d\# q_k$.

Now we turn to proving the {\sf Mer} case itself. In this case, there exist $X:t_1 \triangleq s_1\in S_0$, $Y:t_2
\triangleq s_2\in S_0$, and $\ppi$ such that $\nabla\vdash \ppi\permef
t_1 \approx t_2$ and $\nabla\vdash \ppi\permef s_1 \approx
s_2$. Moreover, by the construction of the derivation, $X:t_1
\triangleq s_1$ is either retained in $S_n$, or is removed from
there because there exist an AUT $Z:t_n\triangleq s_n\in S_n$ and a
permutation $\prho$ such that $\nabla \vdash \prho \permef t_n \approx
t_1$, $\nabla \vdash \prho \permef t_n \approx s_1$, and
$X\vartheta_1^n = \prho\permap Z $. We can turn these two cases into
one, permitting $Z=X$, $t_n=t_1$, $s_n=s_1$, and $\prho = \id$ to
cover also the first case.

Therefore, we can say that there exists a AUT $Z:t_n\triangleq s_n\in
S_n$ such that for some permutation $\prho$, $X\vartheta_1^n=
X\vartheta_2^n= \prho\permap Z$, $Y\vartheta_1^n= \ppi\permef
X\vartheta_2^n = \ppi \prho \permap Z$, $\Gamma_n\setminus
\Gamma_0=\{(\ppi \prho)^{-1}a\# Z \mid a\# Y\in \Gamma_0\}$,
$\nabla\vdash \ppi \prho \permef t_n\approx t_2$, $\nabla\vdash
\ppi\prho \permef s_n\approx s_2$, $\nabla \vdash \prho
\permef t_n \approx t_1$, and $\nabla \vdash \prho \permef s_n \approx
s_1$.

We want to prove $\np{\nablag_n\setminus
  \nablag_0}{Z_0\vartheta_1^n}\preceq \np{\nabla}{t_0}$. First, we
take $\sigma$ such that $Z\sigma=t_n$ and show $(\Gamma_n\setminus
\Gamma_0)\sigma\subseteq \nabla$. For this, we try to prove
$\{b\#U\}\sigma \subseteq \nabla$ for all $b\#U\in \Gamma_n\setminus
\Gamma_0$. By the IH, we have $\{b\#U\}\sigma \subseteq \nabla$ for
all $b\#U\in \Gamma_n\setminus \Gamma_1$. Note that $Z\sigma=t_n$ does
not restrict generality, because if $U=Z$, then by the proposition we proved at the beginning of the {\sf Mer} case
we have that $b\#U\in \Gamma_n\setminus \Gamma_1$ implies
$\nabla\vdash b\# t_n$. 

Therefore, $\{b\#U\}\sigma ={\fe}(\{b\#
Z\sigma\}) = {\fe}(\{b\# t_n\})$ and by
Theorem~\ref{thm:properties:of:FC} we indeed have $\{b\#U\}\sigma
\subseteq \nabla$. Now assume $b\#U\in (\Gamma_n \setminus \Gamma_0)
\setminus (\Gamma_n\setminus \Gamma_1)$. Then $b\# U\in \Gamma_n\cap
(\Gamma_1\setminus \Gamma_0)$. That means, $b\# U= \ppi^{-1}\permef a
\# X$, where $a\# Y \in \Gamma_0$. Moreover, the AUT $X:t_1\triangleq
s_1$ has been retained in $S_n$. From the latter we have, in fact,
$Z=X$, $t_n=t_1$, and $s_n=s_1$. Then $\{b\#U\}\sigma =
{\fe}(\{\ppi^{-1}\permef a \# X\sigma \mid a\# Y \in \Gamma_0\}) =
{\fe}(\{a \# \ppi \permef t_1 \mid a\# Y \in \Gamma_0\})$. On the
other hand, from the assumption on $P_0;S_0;\Gamma_0;\vartheta_0$ we know that for all $a \# Y\in
\Gamma_0$ we have $\nabla \vdash a \# t_2$,
from which by $\nabla
\vdash \ppi \permef t_1 \approx t_2$ we get $\nabla \vdash a \# \ppi
\permef t_1$. Hence, we can apply Theorem~\ref{thm:properties:of:FC}
to ${\fe}(\{a \# \ppi \permef t_1 \mid a\# Y \in \Gamma_0\})$,
obtaining $\{b\#U\}\sigma\subseteq \nabla$ also in this case. Hence,
$(\Gamma_n\setminus \Gamma_0)\sigma\subseteq \nabla$.

It remains to prove $\nabla \vdash Z_0\vartheta_1^n\sigma \approx
t_0$. First, assume $Z_0=X$, $t_0=t_1$, $s_0=s_1$. Then we have
$\nabla \vdash Z_0\vartheta_2^n\sigma \approx t_0$, because
$Z_0\vartheta_2^n\sigma = \prho\permef Z\sigma = \prho \permef t_n$
and we know that $\nabla \vdash \prho \permef t_n \approx t_1$. Since
$Z_0\vartheta_2^n=Z_0\vartheta_1^n$, we get $\nabla \vdash
Z_0\vartheta_1^n\sigma \approx t_0$. Hence, we proved
$\np{\nablag_n\setminus \nablag_0}{Z_0\vartheta_1^n}\preceq
\np{\nabla}{t_0}$ for this case. $\np{\nablag_n\setminus
  \nablag_0}{Z_0\vartheta_1^n}\preceq \np{\nabla}{s_0}$ can be proved
similarly.

Now let $Z_0=Y$, $t_0=t_2$, $s_0=s_2$ and prove again $\nabla \vdash
Z_0\vartheta_1^n\sigma \approx t_0$. Then $Z_0\vartheta_1^n\sigma =
\ppi \prho \permef Z \sigma = \ppi \prho \permef
t_n$. But we have already seen that $\nabla \vdash \ppi\prho
\permef t_n \approx t_2$. Hence, $\nabla \vdash Z_0\vartheta_1^n\sigma
\approx t_0$ is proved. It implies $\np{\nablag_n\setminus
  \nablag_0}{Z_0\vartheta_1^n}\preceq \np{\nabla}{t_0}$ for this
case. $\np{\nablag_n\setminus \nablag_0}{Z_0\vartheta_1^n}\preceq
\np{\nabla}{s_0}$ can be proved similarly.
\end{proof}

The Completeness Theorem states that for any given $A$-based generalization of two input terms-in-context, $\frN$ can compute one which is at most as general than the given one.
\begin{restatable}[Completeness of $\frN$]{theorem}{completenessN}
  \label{thm:completeness}
   Given terms $t$ and $s$ and freshness contexts $\nabla$ and~$\Gamma$,
   all based on a finite set of atoms $\Atoms$. If $\np{\Gamma}{r}$ is
   an $\Atoms$-based generalization of $\np{\nabla}{t}$ and
   $\np{\nabla}{s}$, then there exists a derivation $\{X: t\triangleq
   s\};\, \emptyset;\, \emptyset;\, $ $\varepsilon \Lra^+
   \emptyset;\,S;$ $ \Gamma';\, \sigma$ obtained by an execution of $\frN$, such that
   $\np{\Gamma}{r}\preceq \np{\Gamma'}{X\sigma}$.
\end{restatable}

\begin{proof}
 By structural induction on $r$. We can assume without loss of generality that $\np{\Gamma}{r}$ is an lgg of $\np{\nabla}{t}$ and $\np{\nabla}{s}$.

Let $r$ be an atom $a$. Then $t=s=a$. Therefore, the {\sf Dec} rule gives $\np{\emptyset}{a}$ as the computed answer. To show that $\np{\Gamma}{a}\preceq \np{\emptyset}{a}$, it is enough to take a substitution $\sigma$ such that $X\sigma \neq b$ for each $b\# X\in \Gamma$. Note that it is not necessary $b\in \Atoms$.

Let $r$ be an abstraction $c.r'$. Then $t=a.t'$, $s=b.s'$, $c\in A$, $\nabla\vdash c \# t$, $\nabla\vdash c \# s$, and $\np{\Gamma}{r'}$ is an $\Atoms$-based generalization of $\np{\nabla}{t'}$ and $\np{\nabla}{s'}$. In this case, the {\sf Abs} rule can be applied, which gives $\{Y:(c\,a)\permef t'\triangleq (c\,b)\permef s'\};\emptyset;\emptyset;\sigma_1$, where $\sigma_1=\{X\mapsto c.Y\}$. By the induction hypothesis, we can compute $\Gamma'$ and $\sigma_2$ such that $\np{\Gamma}{r'}\preceq \np{\Gamma'}{Y\sigma_2}$. Let $\sigma=\sigma_1\sigma_2$. We get $\np{\Gamma}{r}=\np{\Gamma}{c.r'} \preceq \np{\Gamma'}{c.Y\sigma_2}= \np{\Gamma'}{X\sigma}$.

Let $r$ be a suspension $\ppi\permap Z$. Since $\np{\Gamma}{r}$ is an lgg of $\np{\nabla}{t}$ and $\np{\nabla}{s}$, the context $\Gamma$ contains all constraints $\ppi^{-1}\permef a\# Z$ such that $\nabla\vdash a\# t$ and $\nabla\vdash a\# s$, and the following alternatives are possible:
\begin{itemize}
\item[(a)] $t$ and $s$ have distinct heads: ${ \head(t)}\neq { \head(s)}$, or
\item[(b)] $t$ and $s$ are both suspensions: $ t= \ppi_1 \permap Y_1$ and $ s=\ppi_2\permap Y_2$, where $\ppi_1,\ppi_2$ and $ Y_1,Y_2$ are not necessarily distinct, or
\item[(c)] $t$ and $s$ are abstractions, but $\Atoms$ does not contain an appropriate fresh atom to uniformly rename the bound atoms in $t$ and $s$.
\end{itemize}

These alternatives give exactly the conditions of the {\sf Sol} rule. Hence, we can apply it, getting $\emptyset;\{X:t\triangleq s\};\Gamma';\sigma$, where $\Gamma'=\{a\# X\mid a\in \Atoms\wedge \nabla \vdash a\# t \wedge \nabla\vdash a\# s\}$ and $\sigma=\varepsilon$. Then $\np{\Gamma}{r}\preceq \np{\Gamma'}{X\sigma}$, which can be confirmed by the substitution $\{Z\mapsto \ppi^{-1}\permap X\}$. 

Let $r$ be a term $f(r_1,\ldots,r_n)$. Then $t=f(t_1,\ldots,t_n)$, $s=f(s_1,\ldots,s_n)$, and $\np{\Gamma}{r_i}$ is a generalization of $\np{\nabla}{t_i}$ and $\np{\nabla}{s_i}$. We proceed by the {\sf Dec} rule, obtaining $\{Y_i : t_i\triangleq s_i \mid 1\le i\le n\};\emptyset;\emptyset;\{X\mapsto f(Y_1,\ldots,Y_n)\}$. By the induction hypothesis, we can construct derivations $D_1,\ldots,D_n$ computing the substitutions $\sigma_1,\ldots,\sigma_n$, respectively, such that $\np{\Gamma}{r_i}\preceq \np{\Gamma'_i}{Y_i\sigma_i}$ for $1\le i \le n$. We combine these derivations, together with the initial \textsf{Dec} step, into one derivation of the form ${\sf D} = \{X:t \triangleq s \};S_0;\Gamma'_0;\sigma_0\Lra \{Y_i:t_i \triangleq s_i \mid 1\le i \le n\};S_1;\Gamma'_1;\sigma_0\sigma_1 \Lra^* \emptyset; S_n; \Gamma'_n; \sigma_0\sigma_1\cdots\sigma_n$, where $\Gamma'_0=\Gamma'_1=\emptyset$, $\sigma_0=\varepsilon$, and $\sigma_1= \{X\mapsto f(Y_1,\ldots,Y_n)\}$. If $r$ does not contain the same variable more than once, $\np{\Gamma}{r_i}\preceq \np{\Gamma'_i}{Y_i\sigma_i}$ for all $1\le i \le n$ imply $\np{\Gamma}{r} = \np{\Gamma}{f(r_1,\ldots,r_n)} \preceq \np{\Gamma'}{f(Y_1,\ldots,Y_n)} = \np{\Gamma'}{X\sigma}$. If $r$ contains the same variable at positions $\pos_1$ and $\pos_2$ (in subterms of the form $\ppi_1\permap Z$ and $\ppi_2\permap Z$), it indicates that
\begin{enumerate}[(a)]
  \item\label{a} the path to $\pos_1$ is the same (modulo bound atom renaming) in $t$ and $s$. It equals (modulo bound atom renaming) the path to $\pos_1$ in $r$, and
  \item\label{b} the path to $\pos_2$ is the same (modulo bound atom renaming) in $t$ and $s$. It equals (modulo bound atom renaming) the path to $\pos_2$ in $r$.
  \item\label{c} there exists a substitution $\vartheta_1$, which respects $\Gamma$, such that $\Gamma \vdash \ppi_1\permap Z\vartheta_1 \approx \ptau_1\permef t|_{\pos_1}$ and $\Gamma \vdash \ppi_2\permef Z\vartheta_1 \approx \ptau_2\permef t|_{\pos_2}$, where $\ptau_1$ and $\ptau_2$ are permutations which rename atoms bound in $t$ by fresh ones,
  \item\label{d} there exists a substitution $\vartheta_2$, which respects $\Gamma$, such that $\Gamma \vdash \ppi_1\permef Z\vartheta_2 \approx \prho_1\permef s|_{\pos_1}$ and $\Gamma \vdash \ppi_2\permef Z\vartheta_2 \approx \prho_2\permef s|_{\pos_2}$, where $\prho_1$ and $\prho_2$ are permutations which rename atoms bound in $s$ by fresh ones,
\end{enumerate}

Then, because of (\ref{a}) and (\ref{b}), we should have two AUTs in $S_n$: One, between (renamed variants of) $t|_{\pos_1}$ and $s|_{\pos_1}$, and the other one between (renamed variants of) $t|_{\pos_2}$ and $s|_{\pos_2}$.  The possible renaming of bound atoms is caused by the fact that \textsf{Abs} might have been applied to obtain the AUTs. From (\ref{c}) and (\ref{d}) we know that $\ptau_1,\ptau_2,\prho_1,\prho_2$ are the names of those renaming permutations. Let those AUTs be $Z_1: \ptau_1\permef t|_{\pos_1} \triangleq \prho_1\permef s|_{\pos_1}$ and $Z_2 : \ptau_2\permef t|_{\pos_2} \triangleq \prho_2\permef s|_{\pos_2}$.

From (\ref{c}) we get $\Gamma \vdash  Z\vartheta_1 \approx \ppi_1^{-1} \ptau_1\permef t|_{\pos_1}$ and $\Gamma \vdash Z\vartheta_1 \approx \ppi_2^{-1} \ptau_2\permef t|_{\pos_2}$, which imply $\Gamma \vdash  \ppi_1^{-1} \ptau_1\permef t|_{\pos_1} \approx \ppi_2^{-1} \ptau_2\permef t|_{\pos_2} $ and, finally, $\Gamma \vdash \ppi_2 \ppi_1^{-1}\permef \ptau_1\permef t|_{\pos_1} \approx \ptau_2\permef t|_{\pos_2} $. Similarly, from (\ref{d}) we get $\Gamma \vdash  \ppi_2 \ppi_1^{-1} \prho_1\permef s|_{\pos_1} \approx \prho_2\permef s|_{\pos_2} $.

That means, we can make the step with the \textsf{Mer} rule for $Z_1: \ptau_1\permef t|_{\pos_1} \triangleq \prho_1\permef s|_{\pos_1}$ and $Z_2 : \ptau_2\permef t|_{\pos_2} \triangleq \prho_2\permef s|_{\pos_2}$ with the substitution $\sigma'_1=\{Z_2\mapsto \ppi_2 \ppi_1^{-1}\permap Z_1\}$.
We can repeat this process for all duplicated vari\-ables in $r$, extending $\sf D$ to the derivation
$\{X :t \triangleq \penalty10000 s \};S_0;\allowbreak \Gamma'_0; \sigma_0\Lra \{Y_i :t_i \triangleq s_i \mid 1\le i \le n\};$ $S_1;\Gamma'_1;\sigma_0 \Lra^* \emptyset; S_n;\Gamma'_n;
\sigma_0\sigma_1\cdots\sigma_n \Lra^+ \emptyset; S_{n+m}; \Gamma'_{n+m}; \allowbreak\sigma_0\sigma_1\cdots\sigma_n\sigma'_1\cdots\sigma'_m$, where $\sigma'_1,\ldots, \sigma'_m$ are substitutions introduced by the applications of the \textsf{Mer} rule. Let $\sigma =\sigma_0\sigma_1\cdots\sigma_n\sigma'_1\cdots\sigma'_m $ and $\Gamma'=\Gamma'_{n+m}$. By this construction, we have $\np{\Gamma}{r} \preceq \np{\Gamma'}{X\sigma}$, which finishes the proof.

\end{proof}

Depending on the selection of AUTs to perform a step, there can be different derivations in $\frN$ starting from the same AUT, leading to different generalizations. The next theorem states that all those generalizations are the same modulo variable renaming and $\alpha$-equivalence.

\begin{restatable}[Uniqueness Modulo $\simeq$]{theorem}{uniquenessN}
  \label{thm:uniqueness}
  Let $t$ and $s$ be terms and $\nabla$ be a freshness context that are based on the same finite set of atoms.
  Let $\{X: t\triangleq s\};\, \emptyset;\, \emptyset;\, $
  $\varepsilon \Lra^+ \emptyset;\,S_1;$ $ \nablag_1;\, \sigma_1$ and
  $\{X: t\triangleq s\};\, \emptyset;\, \emptyset;\, $ $\varepsilon
  \Lra^+ \emptyset;\,S_2;$ $ \nablag_2;\, \sigma_2$ be two maximal
  derivations in $\frN$. Then $\np{\nablag_1}{X\sigma_1}\simeq
  \np{\nablag_2}{X\sigma_2}$.
\end{restatable}

\begin{proof}
It is not hard to notice that if it is possible to change the order of
applications of rules (but sticking to the same selected AUTs for each
rule) then the result remains the same (modulo fresh variable and atom names): Let ${\sf D}_1$ and ${\sf D}_2$ be two two-step derivations ${\sf D}_1 =
P_1;S_1;\Gamma_1;\sigma_1 \Lra_{\mathsf{R}_1} P_2;S_2;\Gamma_2;\sigma_1\vartheta_1
\Lra_{\mathsf{R}_2} P_3; S_3; \Gamma_3; \sigma_1\vartheta_1\vartheta_2$ and
${\sf D}_2= P_1;S_1;\Gamma_1;\sigma_1 \Lra_{\mathsf{R}_2}
P'_2;S'_2;\Gamma'_2;\sigma_1\vartheta_2 \Lra_{\mathsf{R}_1} P'_3; S'_3;\allowbreak\Gamma'_3;
\sigma_1\vartheta_2\vartheta_1$, where
$\mathsf{R}_1$ and $\mathsf{R}_2$ are (not necessarily different) rules and
each of them transforms \emph{exactly the same AUT(s)} in both ${\sf D}_1$ and
${\sf D}_2$. Then these AUT(s) are already present in $P_1\cup S_1$: They are introduced neither by $\mathsf{R}_1$ nor by $\mathsf{R}_2$. Therefore, $\dom(\vartheta_2)\cap \ran(\vartheta_1) = \dom(\vartheta_1)\cap \ran(\vartheta_2) = \emptyset$. Moreover, if we assume that the fresh variables and atoms introduced by the rules are the same in both derivations, then $P_3=P'_3$, $S_3=S'_3$, $\Gamma_3=\Gamma'_3$, and $\sigma_1\vartheta_1\vartheta_2=\sigma_1\vartheta_2\vartheta_1$.

Decomposition, Abstraction, and Solving rules transform the selected AUT in a unique way. We show that it is irrelevant in which order we decide equivariance in the Merging rule.

Let $P ;\; \{Z : t_1 \triangleq s_1, Y : t_2 \triangleq s_2\}\dotcup S ;\; \Gamma;\; \sigma \Lra P ;\; \{Z : t_1 \triangleq s_1\}\dotcup S ;\allowbreak \Gamma \{Y\mapsto \ppi \permap Z\};\; \sigma\{Y \mapsto \ppi \permap Z\} $ be the merging step with $\nabla \vdash \ppi\permef t_1\approx t_2$ and $\nabla \vdash \ppi\permef s_1\approx s_2$. If we do it in the other way around, we would get the step $P ;\; \{Z : t_1 \triangleq s_1, Y : t_2 \triangleq s_2\}\dotcup S ;\; \Gamma;\; \sigma \Lra P ;\; \{Y : t_2 \triangleq s_2\}\dotcup S ;\; \Gamma\{Z\mapsto \ppi^{-1} \permap Y\};\; \sigma\{Z \mapsto \ppi^{-1} \permap Y\} $.

Let $\vartheta_1 = \sigma\varphi_1$ with $\varphi_1=\{Y \mapsto \ppi \permap Z\}$ and $\vartheta_2=\sigma\varphi_2$ with $\varphi_2=\{Z \mapsto \ppi^{-1}\permap Y \}$. Our goal is to prove that $\np{\Gamma\varphi_1}{X\vartheta_1} \simeq \np{\Gamma\varphi_2}{X\vartheta_2}$. For this, we need to prove both $\np{\Gamma\varphi_1}{X\vartheta_1} \preceq \np{\Gamma\varphi_2}{X\vartheta_2}$ and $\np{\Gamma\varphi_2}{X\vartheta_2} \preceq \np{\Gamma\varphi_1}{X\vartheta_1}$.

First, prove $\np{\Gamma\varphi_1}{X\vartheta_1} \preceq \np{\Gamma\varphi_2}{X\vartheta_2}$. We should find such a $\varphi$ that $\Gamma\varphi_1\varphi \subseteq \Gamma\varphi_2$ and $\Gamma\varphi_2 \vdash X\vartheta_1\varphi \approx X\vartheta_2$. 

Take $\varphi=\varphi_2$. Note that for any term $t$, we have $\Gamma\varphi_2 \vdash t \varphi_1\varphi_2 \approx t\varphi_2$, because $\varphi_1\varphi_2 =\{Z \mapsto \ppi \ppi^{-1}\permap Y\}$ and we have $\Gamma\varphi_2 \vdash \ppi \ppi^{-1}\permap Y \approx Y$.
 Therefore, $\Gamma\varphi_2 \vdash X\vartheta_1\varphi \approx X\sigma\varphi_1\varphi_2 \approx X\sigma\varphi_2 \approx X\vartheta_2$ holds.

As for $\Gamma\varphi_1\varphi \subseteq \Gamma\varphi_2$, note that $\varphi_2$ respects $\Gamma\varphi_1$, because it replaces a variable with a suspension and the {\fe} algorithm will have to apply only {\sf Sus-E} rule. We introduce notations $\Gamma_U$ and $\overline{\Gamma}_U$ for any freshness context $\Gamma$ and a variable $U$, denoting $ \Gamma_U:= \{a\# U \mid a\# U \in \Gamma\}$ and $\overline{\Gamma}_U := \Gamma \setminus \Gamma_U $. Then $\Gamma\varphi_1= \overline{\Gamma}_Y \cup \Gamma_Y\varphi_1$ and $\Gamma\varphi_2= \overline{\Gamma}_Z \cup \Gamma_Z\varphi_2$.

Under this notation, $\Gamma\varphi_1\varphi_2=\overline{\Gamma}_Y\varphi_2 \cup \Gamma_Y\varphi_1\varphi_2$. Take $\overline{\Gamma}_Y\varphi_2$. We have $\overline{\Gamma}_Y\varphi_2 =
(\overline{\Gamma}_Y \setminus (\overline{\Gamma}_Y)_Z) \cup ((\overline{\Gamma}_Y)_Z)\varphi_2 = (\overline{\Gamma}_Y \setminus \Gamma_Z) \cup \Gamma_Z\varphi_2$. Since $\Gamma_Z\cap \Gamma_Z\varphi_2=\emptyset$, the we obtain $(\overline{\Gamma}_Y \setminus \Gamma_Z) \cup \Gamma_Z\varphi_2=(\overline{\Gamma}_Y \cup \Gamma_Z\varphi_2) \setminus \Gamma_Z$. As for $\Gamma_Y\varphi_1\varphi_2$, it is easy to see that $\Gamma_Y\varphi_1\varphi_2=\Gamma_Y$.

Hence, we get $\Gamma\varphi_1\varphi_2= ((\overline{\Gamma}_Y \cup \Gamma_Z\varphi_2) \setminus \Gamma_Z) \cup \Gamma_Y$. Since $\Gamma_Z\cap \Gamma_Y=\emptyset$, we get
$((\overline{\Gamma}_Y \cup \Gamma_Z\varphi_2) \setminus \Gamma_Z) \cup \Gamma_Y=((\overline{\Gamma}_Y \cup \Gamma_Z\varphi_2) \cup \Gamma_Y) \setminus \Gamma_Z = (\Gamma \cup \Gamma_Z\varphi_2) \setminus \Gamma_Z$. Since $\Gamma_Z\varphi_2\cap \Gamma_Z=\emptyset$, we get $(\Gamma \cup \Gamma_Z\varphi_2) \setminus \Gamma_Z=(\Gamma \setminus \Gamma_Z) \cup \Gamma_Z\varphi_2$. Hence, $\Gamma\varphi_1\varphi_2= \overline{\Gamma}_Z \cup \Gamma_Z\varphi_2 = \Gamma\varphi_2$. 

We proved $\np{\Gamma\varphi_1}{X\vartheta_1} \preceq \np{\Gamma\varphi_2}{X\vartheta_2}$. With a similar reasoning we can show $\np{\Gamma\varphi_2}{X\vartheta_2} \preceq \np{\Gamma\varphi_1}{X\vartheta_1}$.
\end{proof}

Theorems~\ref{thm:soundness}, \ref{thm:completeness}, and \ref{thm:uniqueness} imply that nominal anti-unification is unitary: For any $A$-based $\nabla$, $t$, and $s$, there exists an $A$-based lgg of $\np{\nabla}{t}$ and $\np{\nabla}{s}$, which is unique modulo $\simeq$ and can be computed by the algorithm $\frN$. 

Now we study how lgg's of terms-in-context depend on the set of atoms the terms-in-context are based on. The following lemma states the precise dependence.

 \begin{lemma}
    \label{lemma:subset:gen}
    Let $A_1$ and $A_2$ be two finite sets of atoms with $A_1\subseteq A_2$
    such that the $A_1$-based terms-in-context $\np{\nabla}{t}$ and
    $\np{\nabla}{s}$ have an $A_1$-based lgg $\np{\Gamma_1}{r_1}$ and
    an $A_2$-based lgg $\np{\Gamma_2}{r_2}$. Then $\Gamma_2\vdash r_1\preceq r_2$.
  \end{lemma}

\begin{proof}
$\np{\Gamma_1}{r_1}$ and $\np{\Gamma_2}{r_2}$ are unique modulo $\simeq$. Let ${\sf D}_i$  be the derivation in $\frN$ that computes $\np{\Gamma_i}{r_i}$, $i=1,2$. The number of atoms in $A_1$ and $A_2$ makes a difference in the rule \textsf{Abs}: If there are not enough atoms in $A_1$, an \textsf{Abs} step in ${\sf D}_2$ is replaced by a \textsf{Sol} step in ${\sf D}_1$. It means that for all positions $\pos$ of $r_1$, $r_2|_\pos$ is also defined. Moreover, there might exist a subterm $r_1|_\pos$, which has a form of suspension, while $r_2|_\pos$ is an abstraction. For such positions, $r_1|_\pos\preceq r_2|_\pos$. For the other positions $\pos'$ of $r_1$, $r_1|_{\pos'}$ and $r_2|_{\pos'}$ may differ only by names of generalization variables or by names of bound atoms.

Another difference might be in the application of \textsf{Sol} in both derivations: It can happen that this rule produces a larger $\Gamma'$ in ${\sf D}_2$ than in ${\sf D}_1$, when transforming the same AUT. 

Hence, if there are positions $\pos_1,\ldots,\pos_n$ in $r_1$ such that $r_1|_{\pos_i}=\ppi_i\permap X$, then there exists a substitution $\varphi_X$ such that $\Gamma_2\vdash \ppi_i \permap X\varphi \approx r_2|_{\pos_i}$, $1\le i \le n $. Taking the union of all $\varphi_X$'s where $X\in \vars(r_1)$, we get $\varphi$ with the property $\Gamma_2 \vdash r_1\varphi \approx r_2$. 
\end{proof}

 Note that, in general, we can not replace $\Gamma_2\vdash r_1\preceq r_2$ with $\Gamma_2\vdash r_1\simeq r_2$ in Lemma~\ref{lemma:subset:gen}. The following example
 illustrates this:
 \begin{example}
  \label{exmp:subset:gen}
   Let $t = a.b$, $s= b.a$, $\nabla =\emptyset$, $A_1=\{a,b\}$, and
   $A_2=\{a,b,c\}$. Then for $\np{\nabla}{t}$ and $\np{\nabla}{s}$,
   $\np{\emptyset}{X}$ is an $A_1$-based lgg and $\np{\{c\#X\}}{c.X}$
   is an $A_2$-based lgg. Obviously, $\{c\#X\} \vdash X \preceq c.X$ but not $\{c\#X\} \vdash c.X \preceq X$.
 \end{example}

This example naturally leads to a question: Under which additional conditions can we have $\Gamma_2\vdash r_1\simeq r_2$ instead of $\Gamma_2\vdash r_1\preceq r_2$ in Lemma~\ref{lemma:subset:gen}? To formalize a possible answer to it, we need some notation.

Let the terms $t,s$ and the freshness context $\nabla$ be based on the same set of atoms $A$. The maximal subset of $A$, \emph{fresh} for $t,s$, and $\nabla$, denoted $\unused(A,t,s,\nabla)$, is defined as $A\setminus (\atoms(t,s) \cup \atoms(\nabla))$.

If $A_1\subseteq A_2$ are two sets of atoms such that $t,s,\nabla$ are at the same time based on both $A_1$ and $A_2$, then $\unused(A_1,t,s,\nabla)\subseteq \unused(A_2,t,s,\nabla)$.

Let $\abscard{t}$ stand for the number of abstraction occurrences in
$t$. $|A|$ stands for the cardinality of the set of atoms $A$. We say that a set of atoms $A$ is \emph{saturated} for $A$-based $t,s$ and $\nabla$, if $|\unused(A,t,s,\nabla)|\ge \min\{\abscard{t},\abscard{s}\}$.

The following lemma answers the question posed above:
\begin{lemma}
  \label{lemma:saturation}
  Under the conditions of Lemma~\ref{lemma:subset:gen}, if $A_1$ is
  saturated for $t,s,\nabla$, then $\Gamma_2\vdash r_1\simeq r_2$. 
\end{lemma}
\begin{proof}
  Let ${\sf D}_i$ be the derivation in $\frN$ that computes
  $\np{\Gamma_i}{r_i}$, $i=1,2$. Note that in each of these
  derivations, the number of \textsf{Abs} steps does not exceed
  $\min\{\abscard{t},\abscard{s}\}$. Since $A_1$ is saturated for
  $t,s,\nabla$ and $A_1\subseteq A_2$, $A_2$ is also saturated for
  $t,s,\nabla$. Hence, whenever an AUT between two abstractions is
  encountered in the derivation ${\sf D}_i$, there is always $c\in
  A_1$ available which satisfies the condition of the \textsf{Abs}
  rule. Therefore, such AU-E's are never transformed by
  \textsf{Sol}. We can assume without loss of generality that the
  sequence of steps in ${\sf D}_1$ and ${\sf D}_2$ are the same. we
  may also assume that we take the same fresh variables, and the same
  atoms from $\unused(A_1,t,s,\nabla)$ in the corresponding steps in
  ${\sf D}_1$ and ${\sf D}_2$. Then the only difference between these
  derivations is in the $\Gamma$'s, caused by the \textsf{Sol} rule
  which might eventually make $\Gamma_2$ larger than $\Gamma_1$. The
  $\sigma$'s computed by the derivations are the same and, therefore,
  $r_1$ and $r_2$ are the same (modulo the assumptions on the variable
  and fresh atom names). Hence, $\Gamma_2 \vdash r_1 \simeq r_2$.
\end{proof}

In other words, this lemma answers the following pragmatic question:
Given $t$, $s$ and $\nabla$, how to choose a set of atoms $A$ so that
(a) $t$, $s$, $\nabla$ are $A$-based and (b) in the $A$-based lgg
$\np{\Gamma}{r}$ of $\np{\nabla}{t}$ and $\np{\nabla}{s}$, the term
$r$ generalizes $s$ and $t$ in the ``best way'', maximally preserving
similarities and uniformly abstracting differences between $s$ and
$t$. The answer is: Besides all the atoms occurring in $t$, $s$, or
$\nabla$, $A$ should contain at least $m$ more atoms, where
$m=\min\{\abscard{t},\abscard{s}\}$.

Besides that, the lemma also gives the condition when the NAU-to-HOPAU translation can be used for solving NAU problems: The set of permitted atoms should be saturated.

\section{Deciding Equivariance}
\label{sect:permuting:matcher}

Computation of $\ppi$ in the condition of the rule {\sf Mer} above
requires an algorithm that solves the following problem: Given nominal
terms $t, s$ and a freshness context $\nabla$, find an
$\atoms(t,s)$-based permutation $\ppi$ such that $\nabla\vdash
\ppi\permef t \approx s$. This is the problem of deciding whether $t$
and $s$ are equivariant with respect to $\nabla$. 
In this Section we describe a rule-based algorithm for this problem,
called $\msys$.  

Note that our problem differs from the problem of equivariant
unification considered in~\cite{DBLP:journals/jar/Cheney10}: We do not
solve unification problems, since we do not allow variable
substitution. We only look for permutations to \emph{decide
  equivariance constructively} and provide a dedicated algorithm for
that.

The algorithm $\msys$ works on tuples of the form $\mproblem ;\, \mnabla
;\, \mran;\, \ppi$ (also called states). $\mproblem$ is a set of
equivariance equations of the form $t\perm s$ where $t,s$ are nominal
terms, $\mnabla$ is a freshness context, and $\mran$ is a finite set
of atoms which are available for computing $\ppi$. The latter holds
the permutation to be returned in case of success.

The algorithm is split into two phases. The first one is a
simplification phase where function applications, abstractions and
suspensions are decomposed as long as possible. The second phase is
the permutation computation, where given a set of equivariance
equations between atoms of the form $a\perm b$ we compute the
permutation which will be returned in case of success.
The rules of the first phase are the following:

\infrule{Dec-E}{Decomposition}
    {\{f(t_1,\ldots,t_m) \perm f(s_1,\ldots,s_m)\}\dotcup \mproblem ;\,
     \mnabla ;\, \mran ;\, \idp \Lra
      \{t_1\perm s_1,\ldots, t_m\perm s_m\}\cup \mproblem ; \mnabla ; \mran ;\idp. }

\infrule{Alp-E}{Alpha Equivalence}
    {\{a.t \perm b.s\}\dotcup \mproblem ;\; \mnabla ;\; \mran ;\; \idp \Lra
    \{\swap{\fresh{c}\,}{a}\permef t \perm \swap{\fresh{c}\,}{b}\permef s \}\cup \mproblem ;\; \mnabla ;\; \mran ;\; \idp, }
        [\noindent where $\fresh{c}$ is a fresh atom of the same sort as $a$ and $b$. ] 

\infrule{Sus-E}{Suspension}
    {\{\msusp_1\permap  X \perm \msusp_2\permap  X\}\dotcup \mproblem ;\; \mnabla ;\; \mran ;\; \idp \Lra
     \{\msusp_1\permef a\perm \msusp_2\permef a\mid a\in\mran\wedge a\#X\not\in\mnabla \}\allowbreak \cup \mproblem ;\; \mnabla ;\;\mran ;\;\idp. }
\par\bigbreak

\noindent The rules of the second phase are the following:

\infrule{Rem-E}{Remove}
    { \{a \perm b\} \dotcup \mproblem ;\; \mnabla ;\; \mran ;\;\ppi \Lra 
     \mproblem ;\; \mnabla ;\; \mran\setminus\{b\} ;\; \ppi, }
    [\noindent if $\ppi\permef a=b$.] 

\infrule{Sol-E}{Solve}
    { \{a \perm b\} \dotcup \mproblem ;\; \mnabla ;\; \mran ;\;\ppi \Lra
     \mproblem ;\; \mnabla ;\; \mran\setminus\{b\} ;\; \swap{\ppi\permef a\;}{b}\ppi ,}
    [\noindent if $\ppi\permef a,\, b \in \mran$ and $\ppi\permef a\neq b$.] 

\par\bigbreak

Note that in {\sf Alp-E}, $\fresh{c}$ is fresh means that
$\fresh{c}\notin \mran$ and, therefore, $\fresh{c}$ will not appear in
$\ppi$. These atoms are an auxiliary means which play a role during
the computation but do not appear in the final result.  

Given nominal terms $t, s$, freshness context $\nabla$, we construct a
state $\{t \perm s\};\;\nabla;\allowbreak\; \atoms(t,\allowbreak s);\;\idp$. We will prove
that when the rules transform this state into
$\emptyset;\;\mnabla;\;A;\;\ppi$, then $\ppi$ is an
$\atoms(t,s)$-based permutation such that $\nabla\vdash \ppi\permef t
\approx s$. When no rule is applicable, and the set of equations is
not empty, we will also prove that there is no solution, hence we fail
and return $\bot$.

\begin{example}\label{exmp:equiv}
  We illustrate the algorithm $\msys$ on examples:
\begin{itemize}
\item
  Consider the equivariance problem $E= \{a \perm a, $ $a.\swap{a}{b}(c\,d)\permap X\perm 
\allowbreak b.X\}$ and $\nablai=\{a\#X\}$:
 \begin{alignat*}{3}
   &  \{a \perm a, a.\swap{a}{b}(c\;d)\permap X\perm\allowbreak b.X\} ; \{a\#X\} ;\; \{a,b,c,d\} ;\;\idp  \Lra_{\text{ Alp-E}}{} \\
  &  \{a \perm a, \swap{\fresh{e}\,}{a}\swap{a}{b}(c\;d)\permap X\perm\allowbreak \swap{\fresh{e}\,}{b}\permap X\} ;\{a\#X\} ;\; \{a,b,c,d\} ;\; \idp \Lra_{\text{ Sus-E}} \\
  & \{a \perm a,\fresh{e}\perm\fresh{e},c\perm d, d\perm c \} ;\{a\#X\} ;\; \{a,b,c,d\} ;\; \idp \Lra_{\text{ Rem-E}}{}\\
  & \{\fresh{e}\perm\fresh{e},c\perm d, d\perm c \}; \{a\#X\} ;\; \{b,c,d\} ;\;\idp \Lra_{\text{ Rem-E}}{}\\
  & \{c\perm d, d\perm c \};\{a\#X\} ;\; \{b,c,d\} ;\;\idp \Lra_{\text{ Sol-E}}\\
  & \{d\perm c \} ;\; \{a\#X\} ;\; \{b,c\} ;\; \swap{c}{d} \Lra_{\text{ Rem-E}}\\
  &  \emptyset ;\; \{a\#X\} ;\; \{b\} ;\;\swap{c}{d}.
\end{alignat*}
  \item For $E=\{a.f(b,X) \perm b.f(a,X)\}$ and $\nablai=\{a\#X\}$, $\msys$ returns $\bot$.
  \item For $E=\{a.f(b,\swap{a}{b}\permap X) \perm b.f(a,X)\} $ and $\nablai=\{a\#X\}$, $\msys$ returns $(b\,a)$.
  \item For $E=\{a.b.\swap{a}{b}\swap{a}{c}\permap X =  b.a. \swap{a}{c}\permap X\}$ and $\nabla=\emptyset$, $\msys$ returns $\idp$.
  \item For $E=\{a.b.\swap{a}{b}\swap{a}{c}\permap X =  a.b. (b\, c)\permap X\}$ and $\nabla=\emptyset$, $\msys$ returns $\bot$.
  \end{itemize}
\end{example}

The Soundness Theorem for $\msys$ states that, indeed, the permutation
the algorithm computes shows that the input terms are equivariant:

\begin{restatable}[Soundness of $\msys$]{theorem}{soundnessE} \label{thm:soundness:equiv}
Let $ \{t\perm s\} ;\,\allowbreak \nablai ;\,\allowbreak \Atoms
;\,\allowbreak \idp \Lra^* \emptyset ;\, \nablai;\, B ;\, \ppi$ be a
derivation in $\msys$, then $\ppi$ is an $\Atoms$-based permutation
such that $\nablai\vdash \ppi\permef t\approx s$.
\end{restatable}

\begin{proof}
We assume the success state with $\ppi$ being the computed
permutation.  Since {\sf Sol-E} is the only rule which adds a new
swapping to $\ppi$ and the swapped atoms are required to be from
$\Atoms$, $\ppi$ is $\Atoms$-based.

The proof is by induction on the length of the derivation, and then,
by case analysis on the applied rule. Let $\Gamma$ be the freshness
environment containing all statements $\fresh{c}\#X$ form by a fresh atom
$\fresh{c}$ introduced along all the derivation and a variable $X$ of the
initial equation.

For any transformation step $ E ;\,\allowbreak \nablai ;\,\allowbreak
\Atoms ;\,\allowbreak \ppi \Lra E' ;\, \nablai;\, \Atoms' ;\, \ppi'$
we will prove that if $\nablai\cup\Gamma\vdash\ppi'\permef t_i'\approx
s_i'$, for any $t_i'\perm s_i'\in E'$, then
$\nablai\cup\Gamma\vdash\ppi\permef t_i\approx s_i$ for any $t_i\perm
s_i\in E$, for any possible applied rule. By induction, we will have
$\nablai\cup\Gamma\vdash\ppi\permef t\approx s$ for the initial
equivariance equation $t\perm s$.  Since $\Gamma$ is not relevant to
prove $t\approx s$, we have also $\nablai\vdash\ppi\permef t\approx
s$.

Soundness of {\sf Dec-E}: From $\nablai\vdash
\pi\permef t_1\approx s_1,\dots, \nablai\vdash \pi\permef t_n\approx
s_n$, follows directly, by the theory of alpha-equivalence
$\nablai\vdash f(\pi\permef t_1,\cdots,\pi\permef t_n) \approx
f(s_1,\cdots,s_n)$ and by the rule of swapping application
$\pi\permef f(t_1,\cdots,t_n)=f(\pi\permef t_1,\cdots,\pi\permef
t_n)$, that $\nablai\vdash \pi\permef f(t_1,\cdots,t_n) \approx
f(s_1,\cdots,s_n)$. In this case the permutation, the set of atoms and
the freshness context are not transformed by the rule.

Soundness of {\sf Alp-E}:  Let $\nabla$ be a
freshness context containing $\Gamma$, in particular $\fresh{c}\#X$
for any variable $X\in\vars(t,s)$. Assume $\nabla\vdash \pi(a\,
\fresh{c})\permef t \approx (b\, \fresh{c})\permef s$ by induction
hypothesis. From this, using $\approx$-abs-1 and the fact that $\pi$
does not affect to $\fresh{c}$, we can deduce $\nabla\vdash\pi\permef
c.(a\, \fresh{c})\permef t \approx c.(b\, \fresh{c})\permef s$.  We
can also construct a proof for $\nabla\vdash \fresh{c}\# t$ and
$\nabla\vdash\fresh{c}\# s$.  Therefore, using $\approx$-abs-2, we can
deduce $\nabla\vdash \fresh{c}.(a\, \fresh{c})\permef t \approx a.t$
and $\nabla\vdash \fresh{c}.(b\, \fresh{c})\permef s \approx b.s$.
Now using the lemmas about the transitivity of $\approx$ and additivity
of permutation application:
$$
\begin{array}{l}
\mbox{\it If $\nabla\vdash t\approx s$ and $\nabla\vdash s\approx u$ then $\nabla\vdash t\approx u$}\\
\mbox{\it If $\nabla\vdash t\approx s$ then $\nabla\vdash \pi\permef t\approx \pi\permef s$}
\end{array}
$$ 
we can deduce $\nabla\vdash \pi\permef (a.t) \approx b.s$.  This
proof proves the soundness of {\sf Alp-E}. Notice that $\pi$ does not
change in this rule.

Soundness of {\sf Sus-E}:  By induction hypothesis, assume
$\nabla\vdash\pi\,\pi_1\permef a \approx \pi_2\permef a$, for any atom
$a$ such that $a\in\Atoms$ and $a\#X\not\in\nabla$.  Assume also
$\Gamma\subset\nabla$, hence, for all fresh atoms, we have
$\fresh{c}\#X\in\nabla$. The rest of atoms $b$ are not fresh and
satisfy $b\not\in\Atoms$ and $b\#X\not\in\nabla$. Since $\pi_1$ and
$\pi_2$ only affect to atoms from $A$ or fresh\footnote{Notice that
  $\pi_1$ and $\pi_2$ can only contain swappings of the original
  equation (i.e. $\Atoms$-based) and swappings introduced by {\sf
    Alp-E}.}, and $\pi$ is $\Atoms$-based, we have $\pi\,\pi_1\permef
b = \pi_2\permef b = b$.  Therefore, $\nabla\vdash\pi\,\pi_1\permef a
\approx \pi_2\permef a$, for any atom $a\#X\not\in\nabla$, and by
$\approx$-susp we deduce $\nabla\vdash \pi\,\pi_1\permap X \approx
\pi_2\permap X$.

In the second phase we have to take into account that, in all
derivations of the form $\mproblem ;\; \mnabla ;\; A ;\; \ppi \Lra^*
\emptyset;\; \mnabla ;\; B ;\;\ppi'\ppi$, permutation $\ppi'$ only
affects to atoms from $A$. This can be proved by inspection of the
rules.

Soundness of {\sf Rem-E}: Let be the complete derivation as follows
$$
\begin{array}{l}
\{a \perm b\} \dotcup \mproblem ;\; \mnabla ;\; \mran ;\;\ppi \Lra\\
\mproblem ;\; \mnabla ;\; \mran\setminus\{b\} ;\; \ppi \Lra^*\\
\emptyset;\; \mnabla ;\; B ;\;\ppi'\ppi
\end{array}
$$ 
By induction hypothesis, $\ppi'\ppi$ solves $\mproblem$.  Since the
rule has been applied we also have $\ppi\permef a = b$.  Now, the
property above proves $\ppi'\permef b = b$, since
$b\not\in\mran\setminus\{b\}$. Therefore $\ppi'\ppi\permef a = b$.

Soundness of {\sf Sol-E}: let the derivation be:
$$
\begin{array}{l}
\{a \perm b\} \dotcup \mproblem ;\; \mnabla ;\; \mran ;\;\ppi \Lra\\
\mproblem ;\; \mnabla ;\; \mran\setminus\{b\} ;\; (\ppi\permef a\; b)\ppi \Lra^*\\
\emptyset;\; \mnabla ;\; B ;\;\ppi'(\ppi\permef a\; b)\ppi
\end{array}
$$ By induction hypothesis, $\ppi'(\ppi\permef a\; b)\ppi$ solves
$\mproblem$. Since $b\not\in\mran\setminus\{b\}$, we have $\ppi'\permef b
= b$. Hence $\ppi'(\ppi\permef a\; b)\ppi \permef a = \ppi'\permef b =
b$, and the computed permutation also solves the equivariance equation
$a \perm b$.
\end{proof}

We now prove an invariant lemma that is used in the proof of
completeness Theorem~\ref{thm:completeness:equiv}.

\begin{lemma}[Invariant Lemma]
  \label{lem:invariant}
  Let $A$ be a finite set of atoms, $E_1$ be a set of equivariance
  equations for terms based on $A$, $\ppi_1$ be an $A$-based
  permutation and $A_1\subseteq A$.  Let $E_1;\nabla;A_1;\ppi_1\Lra
  E_2;\nabla;A_2;\ppi_2$ be any step performed by a rule in $\msys$.
  Let $\Gamma = \{\fresh{c}\#X\mid X\in\vars(E_1),\, \mbox{$\fresh{c}$ is
    a fresh variable}\}$.  Let $\pmu$ be an $A$-based permutation such
  that $\nabla\cup\Gamma \vdash \pmu\permef t \approx s$, for all $t\perm s\in
  E_1$. Then
\begin{enumerate}
  \item\label{item:inv:1} $\nabla\cup\Gamma  \vdash \pmu \permef t'
    \approx s' $, for all $t'\perm
    s'\in E_2$.
  \item\label{item:inv:2} If $\pmu^{-1} \permef b = \ppi_1^{-1}
    \permef b$, for all $b \in A \setminus A_1$, then $\pmu^{-1}
    \permef b = \ppi_2^{-1} \permef b$, for all $b \in A \setminus
    A_2$.
\end{enumerate}
\end{lemma}
\begin{proof} 
  By case distinction on the applied rule.

  {\sf Dec-E}: The proposition is obvious.

  {\sf Alp-E}: In this case it follows from the definitions of
  $\approx$ and permutation application.

  {\sf Sus-E}: In this case $t=\ptau_1\permap X$, $s=\ptau_2\permap
  X$, and by the assumption we have $\nabla \vdash \pmu \ptau_1\permap
  X \approx \ptau_2\permap X $. By the definition of $\approx$, it
  means that we have $a\#X\in\nabla$, for all atoms $a$ such that
  $\pmu \ptau_1 \permef a \neq \ptau_2 \permef a$. Hence, for all
  $a\in\Atoms$ with $a\#X\notin\nabla$ we have $\nabla \vdash \pmu
  \ptau_1 \permef a \approx \ptau_2 \permef a$. This implies that
  $\pmu$ also solves the equations in $\mproblem_2$, hence
  item~\ref{item:inv:1} of the lemma.

  Item~\ref{item:inv:2} of the lemma is trivial for these three
  rules, since $A_1=A_2$ and $\ppi_1=\ppi_2=\idp$.

  {\sf Rem-E}: The item~\ref{item:inv:1} is trivial. To prove the
  item~\ref{item:inv:2}, note that $t=a$, $s= b$, $\ppi_1=\ppi_2$ and
  we only need to show $\pmu^{-1}\permef b = \ppi_2^{-1} \permef
  b$. By the assumption we have $\nabla \vdash \pmu \permef a
  \approx b $. Since $a$ and $b$ are atoms, the latter simply means
  that $\pmu\permef a = b$. From the rule condition we also know that
  $\ppi_1\permef a = b$. From these two equalities we get $\pmu^{-1}
  \permef b = a = \ppi_2^{-1} b$.

  {\sf Sol-E}: The item~\ref{item:inv:1} is trivial also in this case.
  To prove the item~\ref{item:inv:2}, note that $t=a$, $s= b$,
  $\ppi_2=(\ppi_1\permef a\; b) \ppi_1$ and we only need to show
  $\pmu^{-1}\permef b = \ppi_2^{-1} \permef b$. By the assumption we
  have $\nabla \vdash \pmu \permef a \approx b $, which means that
  $\pmu \permef a = b$ and, hence, $a = \pmu^{-1}\permef b$. As for
  $\ppi_2^{-1}\permef b$, we have $\ppi_2^{-1} \permef b =
  \ppi_1^{-1}(\ppi_1\permef a\;b) \permef b =
  \ppi_1^{-1} \permef (\ppi_1\permef a) = a$. Hence, we get $\pmu^{-1}\permef b = a = \ppi_2^{-1}\permef b$.
\end{proof}

\begin{restatable}[Completeness of $\msys$]{theorem}{completenessE}
\label{thm:completeness:equiv}
   Let $\Atoms$ be a finite set of atoms, $t,s$ be $\Atoms$-based
   terms, and $\nabla$ be a freshness context. If
   $\nablai\vdash \pmu \permef t\approx s$ holds for some
   $\Atoms$-based permutation $\pmu$, then there exists a derivation $
   \{t\perm s\} ;\,\allowbreak \nablai ;\,\allowbreak
   \Atoms;\,\allowbreak \idp \Lra^* \emptyset ;\, \nablag;\, B ;\,
   \ppi $, obtained by an execution of $\msys$, such that $\ppi\permef
   a = \pmu\permef a$ for any atom $a\in\fatom(t)$.
\end{restatable}

\begin{proof}
  First show that under the conditions of the theorem, if $ \{t\perm
  s\} ;\,\allowbreak \nablai ;\,\allowbreak \Atoms ;\,\allowbreak \idp
  \Lra^* \emptyset ;\, \nablag;\, B ;\, \ppi $ is a derivation
  obtained by an execution of $\msys$, then $\ppi\permef a =
  \pmu\permef a$, for any atom $a\in\fatom(t)$. Afterwards we prove
  that (under the conditions of the theorem) there is no failing
  derivation with the rules of $\msys$ starting from $\{t\perm s\}
  ;\,\allowbreak \nablai ;\,\allowbreak \Atoms ;\,\allowbreak
  \idp$. Since all derivations are finite, it will imply the existence
  of $ \{t\perm s\} ;\,\allowbreak \nablai ;\,\allowbreak \Atoms
  ;\,\allowbreak \idp \Lra^* \emptyset ;\, \nablag;\, B ;\, \ppi $.

  Let $\{t\perm s\} ;\allowbreak \nablai ;\allowbreak \Atoms
  ;\allowbreak \idp \Lra^* E'; \nablag'; B' ;\allowbreak \ppi' \Lra^*
  \emptyset ; \nablag; B ;\ppi $ be a derivation, where $E' ;\,
  \nablag';\, B' ;\, \ppi'$ is the first state in the second phase of
  the algorithm. It means that $E'$ contains equations between atoms
  only, and the atoms of $t$ (except, maybe, some bound ones which
  disappear after the application of the {\sf Alp-E} rule) appear in
  the left hand sides of equations in $E$. By
  Lemma~\ref{lem:invariant}, $\nablag' \vdash \pmu \permef a_1 \approx
  a_2 $, for all $a_1\perm a_2\in E'$. By
  Theorem~\ref{thm:soundness:equiv} and Lemma~\ref{lem:invariant} the
  same is true for $\ppi$. Therefore, $\nablag' \vdash \pmu \permef
  a_1 \approx \ppi \permef a_1 $, for all $a_1\perm a_2\in E'$. For
  atoms, $\nabla \vdash a\approx b$ iff $a=b$. Hence, we get
  $\pmu\permef a_1 = \ppi\permef a_1$, for all $a_1\in S$, where
  $\fatom(t)\subseteq S \subseteq \atoms(t)$. It proves $\ppi\permef a=
  \pmu\permef a$, for all $a\in\fatom(t)$, when the desired successful
  derivation exists.

  Now we show that no derivation with the rules of $\msys$ starting
  from $ \{t\perm s\} ;\,\allowbreak \nablai
  ;\,\allowbreak \Atoms ;\,\allowbreak \idp$ fails. Assume by
  contradiction that there exists such a failing derivation. Let
  $E';\nabla';A';\ppi'$ be the final state in it, to which no rule
  applies. Analyzing the rules in $\msys$, one can easily conclude
  that it can be caused by one of the following two cases:
\begin{enumerate}
  \item $E'$ contains an equivariance equation of the form
    $f(t_1,\ldots,t_n)\perm g(s_1,\ldots,\allowbreak s_m)$, where $f\neq g$.
  \item $E'$ contains an equivariance equation of the form $a\perm b$,
    where $\ppi'\permef a \neq b$, such that $\ppi'\permef
    a\notin A'$ or $b \notin A'$.
\end{enumerate} 
   In the first case, by Lemma~\ref{lem:invariant} $\nabla \vdash \pmu
   \permef f(t_1,\ldots,t_n) \approx g(s_1,\ldots,s_m)$ should hold,
   but $f\neq g$ forbids it. Hence, this case is impossible.

   Now we analyze the second case. Consider each condition.

   Condition 1: $\ppi'\permef a\notin A'$. Then either
   $\ppi'\permef a$ is a fresh atom, or $\ppi'\permef a
   \in A\setminus A'$.
\begin{itemize}
  \item \emph{$\ppi'\permef a$ is a fresh atom:} Since $\ppi'$ does
    not affect fresh atoms, we get $a \neq b$. On the other hand, we
    have $\nablag' \vdash \pmu \permef a \approx b$ and, hence, $\pmu
    \permef a = b$, because $\pmu \permef a$ and $b$ are atoms.  Since
    $\pmu$ is $A$-based, $b\notin A$ implies $a=b$. A contradiction.
  \item \emph{$\ppi'\permef a \in A\setminus A'$:} By
    Lemma~\ref{lem:invariant} we get $\pmu^{-1}\ppi' \permef a =
    \ppi'^{-1}\ppi' \permef a = a$. Therefore, $\ppi'\permef a
    =\pmu\permef a$ and we get $\pmu \permef a \neq b$, which
    contradicts $\nablag' \vdash \pmu \permef a \approx b$, because
    $\pmu\permef a$ and $b$ are atoms.
\end{itemize}

Condition 2: $b \notin A'$. Then either $b$ is a fresh atom, or $b \in A\setminus A'$.
\begin{itemize}
  \item \emph{$b$ is a fresh atom:} We obtain a contradiction by a
    reasoning similar to the case when $\ppi'\permef a$ is a
    fresh atom.
  \item \emph{$b \in A\setminus A'$:} The atom $b$ has been removed
    from the set of atoms in the derivation earlier either at {\sf
      Sus-E}, {\sf Rem-E}, or {\sf Sol-E} step, which indicates that
    there is $c\in A\setminus A'$ such that $c=\ppi'^{-1}\permef
    b$. Moreover, $c\neq a$. From Lemma~\ref{lem:invariant} we get
    $c=\pmu^{-1} \permef b$ which, together with $c\neq a$, implies
    $\pmu \permef a \neq b$. But it contradicts $\nablag' \vdash
    \pmu\permef a \approx b$.
\end{itemize}

   The obtained contradiction proves that no derivation with the rules
   of $\msys$ starting from $ \{t\perm s\}
   ;\,\allowbreak \nablai ;\,\allowbreak \Atoms ;\,\allowbreak \idp$
   fails.
\end{proof}

\section{Complexity Analysis}

We represent a permutations $\ppi$ as two hash tables. One for the
permutation itself, we call it $T_\ppi$, and one for the inverse of
the permutation, called $T_{\ppi^{-1}}$. The key of a hash tables is
an atom and we associate another atom, the mapping, with it. For
instance the permutation $\ppi = \swap{a}{b}\swap{a}{c}$ is
represented as $T_\ppi=\{a\mapsto c, b\mapsto a, c\mapsto b\}$ and
$T_{\ppi^{-1}}=\{a\mapsto b, b\mapsto c, c\mapsto a\}$.  We write
$T_\ppi(a)$ to obtain from the hash table $T_\ppi$ the atom which is
associated with the key $a$. If no atom is associated with the key $a$
then $T_\ppi(a)$ returns $a$. We write $T_\ppi(a\mapsto b)$, to set
the mapping such that $T_\ppi(a) = b$.  As the set of atoms is small,
we can assume a perfect hash function. It follows, that both defined
operations are done in constant time, leading to constant time
application of a permutation. Swapping application to a permutation
$\swap{a}{b}\ppi$ is also done in constant time in the following way:
Obtain $c=T_{\ppi^{-1}}(a)$ and $d=T_{\ppi^{-1}}(b)$ and perform the
following updates:
\begin{itemize}
\item[(a)]
			$T_\ppi(c\mapsto b)$ and
			$T_\ppi(d\mapsto a)$,
\item[(b)]
			$T_{\ppi^{-1}}(b\mapsto c)$ and
			$T_{\ppi^{-1}}(a\mapsto d)$.
\end{itemize}

We also represent set membership of atoms to a set of atoms~$A$ with a
hash table $\in_A$ from atoms to Booleans such that $\in_A(a)=true$
iff $a\in A$. We also have a list $L_A$ of the atoms representing the
entries of the table such that $\in_A(a)=true$ to easily know all
atoms in $A$.

Finally we also represent set membership of freshness constraints to a freshness environment $\nabla$ with a hash table~$\in_\nabla$. 
\begin{theorem}\label{thm:match:complexity}
 Given a set of equivariance equations $E$, and a freshness context $\mnabla$.
Let $m$ be the size of $\mnabla$, and let $n$ be the size of $E$. The algorithm $\msys$ has $O(n^2 + m)$ time complexity.
\end{theorem}

\begin{proof}
Collecting the atoms from $E$ in a separate set $A$ does not affect the space complexity and can be done in time $O(n)$. The freshness environment $\nabla$ will not be modified by rule applications and membership test in the rule {\sf Sus-E} can be done in constant time.
We only have to construct the corresponding hash tables in time $O(m)$.
We analyze complexity of both phases. 

For the first phase, notice that all rules can be applied only $O(n)$
many times, since {\sf Dec-E} removes two function symbols and {\sf
  Alp-E} two abstraction, and {\sf Sus-E} two suspensions. The
resulting equations after this phase only contain atoms. However,
notice that the size of these equations is not necessarily linear.
Every time we apply {\sf Alp-E} a new swapping is applied to both
subterms.  This swappings may increase the size of suspensions
occurring bellow the abstraction. Since there are $O(n)$ many
suspensions and $O(n)$ many abstractions, the final size of
suspensions is $O(n^2)$. This is the size of the atom equations at the
beginning of the second phase. We can see that the application of
{\sf Dec-E} rule has $O(1)$ time complexity (with the appropriate
representation of equations). 

The application of {\sf Alp-E} rule requires to find a fresh atom not in $A$, this can be done in constant time. Later, a swapping 
has to be applied twice.
Swapping application requires traversing the term hence has $O(n)$
time complexity. The application of {\sf Sus-E} requires to traverse
$L_A$ ($O(n)$) and check for freshness membership in $\in_\nabla$
($O(1)$). Finally it has to add equations like $(\pi_1\permef a\approx
\pi_2)$, this requires to build $T_{\pi_1}$ and $T_{\pi_2}$ that can
be done in $O(n)$ time complexity and allow us to build each equation
in $O(1)$ time.
Summing up, this phase has $O(n^2)$ time complexity.

For the second phase, notice that both rules {\sf Rem-E} and {\sf
  Sol-E} remove an equation and do not introduce any other one. Hence,
potentially having $O(n^2)$ many equations in this phase, these
equations can be applied $O(n^2)$ may times.  We construct a hash
table $T_\ppi$ for $\ppi$ that will be maintained and used by both
rules. Each application has time complexity $O(1)$.
{\sf Rem-E} uses $T_\ppi$ to check for applicability and if it is
applied, it only removes $b$ from $A$, hence updating $\in_A$ (notice that
we do not care about $L_A$ in this second phase of the algorithm).
{\sf Sol-E} uses $\in_A$ and $T_\ppi$ to check for applicability 
and if it is applied, it only removes $b$ from $A$ (hence updating
$\in_A$), and updates $T_\ppi$.
Summing up, this phase maintains the overall $O(n^2)$ time complexity.%
\end{proof}

\begin{theorem}
  \label{thm:complexity}
 The nominal anti-unification algorithm $\frN$ has $O(n^5)$ time complexity and $O(n^4)$ space complexity, where $n$ is the input size.
\end{theorem}

\begin{proof}
By design of the rules and
theorem \ref{thm:uniqueness} we can arrange a maximal derivation like
$\{X_0: t_0 \triangleq s_0 \};\emptyset;\emptyset;\ids
\Lra^*_{\textsf{Dec,Abs,Sol}} \emptyset;S_l;\Gamma_l;\sigma_l
\Lra^*_{\textsf{Mer}} \emptyset;S_m;\Gamma_m; \sigma_m$, postponing
the application of \textsf{Mer} until the end.  Rules {\sf Dec},
{\sf Abs} and {\sf Sol} can be applied $O(n)$ many times. 
However, notice that every application of {\sf Abs} may increase the size of every suspension below.  
Hence, the
size of the store $S_l$ is $O(n^2)$, although it only contain $O(n)$ equations, after an exhaustive derivation
$\{X_0: t_0 \triangleq s_0 \};\emptyset;\emptyset;\ids
\Lra^*_{\textsf{Dec,Abs,Sol}} \emptyset;S_l;\Gamma_l;\sigma_l$.

Now we turn to analyzing the transformation phase
$\emptyset;S_l;\Gamma_l;\sigma_l \Lra^*_{\textsf{Mer}}
\emptyset;S_m;\allowbreak \Gamma_m; \sigma_m$. 
Let $S_l=\{X_1:t_1\triangleq
s_1,\dots,X_k:t_k\triangleq s_k\}$ and $n_i$ be the size of
$X_i:t_i\triangleq s_i$, $1\leq i\leq k$, then
$\sum_{i=1}^{k}n_i=O(n^2)$ and $k=O(n)$. From theorem
\ref{thm:match:complexity} we know that solving the equivariance
problem for two AUPs $X_i:t_i\triangleq s_i$ and $X_j:t_j\triangleq
s_j$ and an arbitrary freshness context $\nabla$ requires $O((n_i+n_j)^2+m)$ time
and space, where $m$ is the size of $\nabla$ with $m=O(n)$.

Merging requires to solve this problem for each pair of AUPs. This
leads to the time complexity
$\sum_{i=1}^{k}\sum_{j=i+1}^{k}O((n_i+n_j)^2+m)\leq
O(\sum_{i=1}^{k}\sum_{j=1}^{k}(n_i+n_j)^2)+O(\sum_{i=1}^{k}\sum_{j=1}^{k}m)$. The
second sum is $\sum_{i=1}^{k}\sum_{j=1}^{k}m=k^2m=O(n^3)$. Now we
estimate an upper bound for the sum
$\sum_{i=1}^{k}\sum_{j=1}^{k}(n_i+n_j)^2=\sum_{i=1}^{k}\sum_{j=1}^{k}n_i^2+\allowbreak\sum_{i=1}^{k}\sum_{j=1}^{k}2n_i
n_j\;+\;\sum_{i=1}^{k}\sum_{j=1}^{k}n_j^2\leq \sum_{i=1}^{k}k
n_i^2\;+\; 2\left(\sum_{i=1}^{k}n_i\right)\left(\sum_{j=1}^{k}n_j\right)+\allowbreak\sum_{i=1}^{k}(\sum_{j=1}^{k}n_j)^2\leq
k(\sum_{i=1}^{k}n_i)^2+2O(n^2)
O(n^2)\;+ \sum_{i=1}^{k}O(n^2)=k O(n^2)^2+2 O(n^2)^2+k O(n^2)^2=O(n^5)$,
resulting into the stated bounds.

The space is bounded by the space required by a single call to the equivariance algorithm
with an imput of size $O(n^2)$, hence $O(n^4)$.
\end{proof}

\section{Conclusion}
\label{sect:conclusion}

The problem of anti-unification for nominal terms-in-context is
sensitive to the set of atoms permitted in generalizations: If this
set is infinite, there is no least general generalization. Otherwise
there exists a unique lgg. If this set is finite and satisfies the
notion of being saturated, defined in the paper, then the lgg retains
the common structure of the input nominal terms maximally.

We illustrated that, similar to some other theories where unification,
generalization, and the subsumption relation are defined, the nominal
terms-in-contexts form a join-meet lattice with respect to the
subsumption relation, where the existence of join is unifiability, and
the meet corresponds to least general generalization.

We designed an anti-unification algorithm for nominal
terms-in-context. It contains a subalgorithm that constructively
decides whether two terms are equivariant with respect to the given
freshness context. We proved termination, soundness, and completeness
of these algorithms, investigated their complexities, and implemented
them. Given a fixed set of atoms $A$, the nominal anti-unification
algorithm computes a least general $A$-based term-in-context
generalization of the given $A$-based terms-in-context, and requires
$O(n^5)$ time and $O(n^4)$ space for that, where $n$ is the size of
the input. The computed lgg is unique modulo $\alpha$-equivalence and
variable renaming.

\section*{Acknowledgment}

This\, research\, has\, been\, partially\, supported\, by\, the\, Spanish\, project\, HeLo
(TIN2012-33042), by the Austrian
Science Fund (FWF) with the project SToUT (P 24087-N18), and by the strategic program ``Innovatives O\"{O} 2010plus'' by the Upper Austrian Government.

\bibliographystyle{abbrv}
\bibliography{anti-unif}

\end{document}